\documentclass[10pt]{article}
\usepackage{amsmath}
\usepackage{amsthm}
\usepackage{amsfonts}
\usepackage[english]{babel}
\usepackage[ansinew]{inputenc}
\usepackage{graphicx}
\usepackage{graphpap}
\usepackage{accents}
\usepackage{psfrag}
\usepackage{xcolor}

\def\Journal#1#2#3#4#5#6{#1, ``#2'', {#3}, {\bf #4}, #5 (#6).}
\def\ChapterBook#1#2#3#4#5{#1, ``#2'', en  {\em #3} (#4), #5.}

\newcommand\underlineb[1]{\underline{#1}}

\def\avY{\underlineb{(\XY)}}

\newcommand{\uwidehat}[1]{%
  \mathpalette\douwidehat{#1}%
}

\def\PRD{Physical Review D}
\def\NC{Nuovo Cimento}

\def\den{\mu}
\def\pre{p}
\def\flux{J}
\def\XX{\mathfrak{X}}
\def\dN{\widetilde{\N}}
\def\dhh{\widetilde{h}}
\def\dgamma{\widetilde{\gamma}}
\def\dellc{\widetilde{\ellc}}
\def\dupn{\widetilde{\upn}}
\def\dX{\widetilde{X}}
\def\delltwo{\widetilde{\elltwo}}
\newcommand\dq{\widetilde{q}}
\def\dg{\widetilde{g}}
\def\dU{\widetilde{U}}

\def\dG{\widetilde{\G}}
\def\dD{\widetilde{{\mathcal D}}}
\def\dz{\widetilde{z}}
\def\dgauge{\widetilde{\gauge}}
\def\dnablao{\stackrel{\circ}{\widetilde{\nabla}}}
\def\bulkden{\rho}
\def\bulkpre{\mathfrak{p}}
\def\bulkflux{\mathfrak{J}}
\def\bulkGrav{\R^H}
\def\Grav{\Upsilon}
\def\bGrav{{\bm{\Grav}}}

\newtheorem{theorem}{Theorem}[section]
\newtheorem{terminology}{Terminology}[section]
\newtheorem{lemma}[theorem]{Lemma}

\newtheorem{definition}[theorem]{Definition}

\theoremstyle{remark}
\newtheorem{remark}[theorem]{Remark}
\numberwithin{equation}{section}
\theoremstyle{plain}

\newtheorem{corollary}[theorem]{Corollary}

\newtheorem{proposition}[theorem]{Proposition}

\def\N{{\mathcal N}}
\def\M{{\mathcal M}}
\def\Rad{\mbox{Rad}}
\def\dim{\mbox{dim}}
\def\Id{\mbox{Id}}

\def\Su{S{}^2_0}
\def\Sd{S{}^0_2}

\def\R{\mathcal R}

\newcommand\unl[1]{{#1}_{\ell}}

\newcommand\Q{Q}
\newcommand\Z{\mbox{tr}_P V}
\def\Vnp{(\iota_n V)}
\def\Vn{\iota_n V}
\def\W{\Vn}
\def\Vnn{\Q_{V}}
\def\X{\overline{V}}
\def\tgamma{{\, \mathcal C}}
\def\eqq{\qquad \Longleftrightarrow \qquad}
\def\wn{\bm{\omega}(\upn)}

\newcommand\elltwo{\ell^{(2)}}

\def\n{\mathfrak{n}}
\newcommand{\lp}{\left(}
\newcommand{\rp}{\right)}

\def\LL{\mathfrak B}
\def\aa{\alpha}
\def\bb{\beta}
\def\ellc{{\bmell}}
\def\bmell{\bm{\ell}}

\def\Riem{{\mbox{Riem}}}
\def\Ricc{{\mbox{Ric}}}
\def\Ricco{{\stackrel{\circ}{\Ricc}}}
\def\Riemo{{\stackrel{\circ}{\mbox{Riem}}}}

\def\D{\mathcal D}
\newcommand\rig{\xi}

\def\metdata{\{\N,\gamma,\ellc,\elltwo\}}
\def\dmetdata{\{\dN,\dgamma,\dellc,\delltwo\}}
\def\dmetdataone{\{\dN,\dgamma,\dellc_1,\delltwo_1\}}
\def\dmetdatatwo{\{\dN,\dgamma,\dellc_2,\delltwo_2\}}
\def\metdatabar{\overline{\{\N,\gamma,\ellc,\elltwo\}}}

\def\metdataone{\{\N,\gamma_1,\ellc{}_1,\elltwo_1\}}
\def\metdatatwo{\{\N, \gamma_2,\ellc{}_2,\elltwo_2\}}
\newcommand\hypdata{\{ \mathcal{N},\gamma,\ellc, \elltwo, \bY\}}
\newcommand\hypdatap{\{ \mathcal{N},\gamma,\ellc, \elltwo, \bY^+\}}
\newcommand\hypdatapm{\{ \mathcal{N},\gamma,\ellc, \elltwo, \bY^{\pm}, \epsilon\}}
\newcommand\hypdatam{\{ \mathcal{N},\gamma,\ellc, \elltwo, \bY^-\}}
\def\bY{\textup{{\textbf Y}}}
\def\Y{\textup{Y}}
\def\upn{\textup{n}}
\def\vecn{\upn}
\def\U{\textup{U}}
\def\bU{\textup{\textbf U}}
\def\F{\textup{F}}
\def\bF{\textup{\textbf{F}}}
\def\sone{\textup{s}}
\def\bsone{\textup{\textbf{s}}}
\def\rone{\textup{r}}
\def\brone{\textup{\textbf{r}}}

\def\FF{\mathfrak F}
\def\FF{\mathcal F}

\def\A{{\mathcal A}}
\def\XX{{\mathfrak X}}
\def\G{{\mathcal G}}
\def\gauge{\zeta}

\def\Tau{{\mathcal T}}
\def\XY{\uwidehat{\overline{\Y}}}
\def\YH{\bY^H}
\def\UH{\bU^H}
\def\bYH{\bY^H}
\def\bUH{\bU^H}
\def\ZY{\mbox{tr}_P \bY}
\def\ZU{\mbox{tr}_P \bU}

\newcounter{mnotecount}

\newcommand{\mnotex}[1]
{\protect{\stepcounter{mnotecount}}$^{\mbox{\footnotesize $\bullet$\themnotecount}}$ 
\marginpar{
\raggedright\tiny\em
$\!\!\!\!\!\!\,\bullet$\themnotecount: #1} }

\def\defi{{\stackrel{\mbox{\tiny {\bf def}}}{\,\, = \,\, }}}
\def\be{\begin{equation*}}
\def\en{\end{equation*}}

\makeatletter
\newcommand{\douwidehat}[2]{%
  \sbox0{$\m@th#1\widehat{\hphantom{#2}}\vphantom{t}$}%
  \sbox2{$t$}%
  \dimen2=\ht0
  \advance\dimen2 -\ht2
  \sbox2{$#2$}%
  \dimen0=\ht0
  \rlap{%
    \raisebox{\dimexpr-\dimen0-\dp2-1pt}[0pt][\dimexpr\dimen2+\dp2]{\box0}%
  }
  {#2}%
}
\makeatother

\def\K{{\mathcal K}}
\def\O{{\mathcal O}}

\def\P{\mathbb{P}}
\def\I{\mathbb{I}}

\def\QY{Q^Y}

\def\hatLn{\overline{\pounds}_{\upn}}
\def\VH{V^H}

\begin{document}

\newcommand{\bm}[1]{\mbox{\boldmath $#1$}}


\def\botl{\bot_{\ell}}
\def\parallell{\parallel_{\ell}}

\def\WW{W{}^{\bot\mbox{\tiny $(2)$}}}

\def\nn{n{}^{\mbox{\tiny $(2)$}}}
\def\ll{\ell^{\mbox{\tiny $(2)$}}}
\def\ellell{(\ell \cdot \ell)}
\def\bmomega{{\bm{\omega}}}
\def\omegac{\underaccent{\check}{\bmomega}}

\def\llprime{\ll{}^{\prime}}
\def\nnprime{\nn{}^{\prime}}

\def\C{{\mathcal C}}

\def\E{{\mathcal E}}

\def\Ein{{\mbox{Ein}}}
\def\gM{g}
\def\gSigma{\gamma}
\def\nablaM{\nabla}
\def\nablab{{\overline{\nabla}}}
\def\Riemb{\overline{R}}
\def\Gamo{{\stackrel{\circ}{\Gamma}}}
\def\nablao{{\stackrel{\circ}{\nabla}}}
\def\nabgam{{\nabla^{(\gamma)}}}

\def\Riemgam{R^{(\gamma)}}
\def\acc{{\mbox{acc}}}

\def\Riemoin{{\stackrel{\circ}{R}}}
\def\Vo{{\stackrel{\circ}{V}}}

\def\H{{\mathcal H}}
\def\w{{\omega}}

\def\Sig{\Sigma}
\def\trP{\mbox{tr}_P \,}
\def\tr{\mbox{tr}\,}
\def\Sb{S_{b}}
\def\Db{{\mathfrak D}_{b}}

\def\SN{\Sigma\cap\mathcal{N}}

\def\Journal#1#2#3#4#5#6{#1, ``#2'', {\em #3} {\bf #4}, #5 (#6).}
\def\JournalPrep#1#2{#1, ``#2'',  In preparation}
\def\Monograph#1#2#3#4{#1, ``#2'', #3 (#4)}
\def\Living#1#2#3#4#5#6#7{#1, #2. Living Rev. Relativity \textbf{#3}, #4 (#5),
URL (cited on #6): 
http://www.livingreviews.org/#7}

\def\JGP{\em J. Geom. Phys.}
\def\JDG{\em J. Diff. Geom.}
\def\CQG{\em Class. Quantum Grav.}
\def\JPA{\em J. Phys. A: Math. Gen.}
\def\PRD{{\em Phys. Rev.} \bm{D}}
\def\GRG{\em Gen. Rel. Grav.}
\def\IJT{\em Int. J. Theor. Phys.}
\def\PR{\em Phys. Rev.}
\def\RMP{\em Rep. Math. Phys.}
\def\MNRAS{\em Mon. Not. Roy. Astr. Soc.}
\def\JMP{\em J. Math. Phys.}
\def\DG{\em Diff. Geom.}
\def\CMP{\em Commun. Math. Phys.}
\def\APP{\em Acta Phys. Polon.}
\def\PRL{\em Phys. Rev. Lett.}
\def\ARAA{\em Ann. Rev. Astron. Astroph.}
\def\ANP{\em Annals Phys.}
\def\AP{\em Ap. J.}
\def\APJL{\em Ap. J. Lett.}
\def\MPL{\em Mod. Phys. Lett.}
\def\PREP{\em Phys. Rep.}
\def\AASF{\em Ann. Acad. Sci. Fennicae}
\def\ZP{\em Z. Phys.}
\def\PNAS{\em Proc. Natl. Acad. Sci. USA}
\def\PLMS{\em Proc. London Math. Soth.}
\def\AIHP{\em Ann. Inst. H. Poincar\'e}
\def\ANYAS{\em Ann. N. Y. Acad. Sci.}
\def\SPJ{\em Sov. Phys. JETP}
\def\PAWBS{\em Preuss. Akad. Wiss. Berlin, Sitzber.}
\def\PPLL{\em Phys. Lett. A }
\def\QJRAS{\em Q. Jl. R. Astr. Soc.}
\def\CR{\em C.R. Acad. Sci. (Paris)}
\def\CP{\em Cahiers de Physique}
\def\NC{\em Nuovo Cimento}
\def\AM{\em Ann. Math.}
\def\APP{\em Acta Physica Polonica}
\def\BAMS{\em Bulletin Amer. Math. Soc}
\def\CPAM{\em Commun. Pure Appl. Math.}
\def\PJM{\em Pacific J. Math.}
\def\ATMP{\em Adv. Theor. Math. Phys.}
\def\PRSA{\em Proc. Roy. Soc. A.}
\def\APPT{\em Ann. Poincar\'e Phys. Theory}
\def\RPM{\em Rep. Math. Phys.}
\def\AHP{\em Annales Henri Poincar\'e}

\def\a{{\alpha}}
\def\h{{h}}

\title{Abstract null geometry, energy-momentum map and applications to the constraint tensor}

\author{Marc Mars\footnote{e-mail: marc@usal.es} \\
  Instituto de F\'{\i}sica Fundamental y Matem\'aticas, IUFFyM\\
Universidad de Salamanca}

\maketitle

\begin{abstract} We introduce and study the notion of null manifold. This is
  a smooth manifold ${\mathcal N}$ endowed with a degenerate metric $\gamma$ with one-dimensional radical at every point. We also define the notion of ruled null manifold, which is a special case of null manifolds. We prove that ruled null manifolds are in one-to-one correspondence with equivalence classes of null metric hypersurface data. This correspondence is used to endow any null manifold
  $({\mathcal N},\gamma)$ with a family of torsion-free connections related to each other by a well-defined gauge group. The whole construction allows one to define and use geometric notions on arbitrary null manifolds. The paper has a second part where we introduce a canonical map on any null metric hypersurface data and use its algebraic properties to define a canonical decomposition of any symmetric (0,2)-covariant tensor. This decomposition, together with two new differential operators compatible with this splitting, are used to decompose the constraint tensor in full generality and at the purely abstract level.  This leads to a hierarchical structure of the (detached) Einstein vacuum null constraint equations without the need of introducing special coordinates or special foliations.
  The results are applied to study null shells arising from the matching of two spacetimes across null boundaries. The equations governing such objects are obtained in hierarchical form without imposing any topological, gauge or coordinate conditions on the shell.
\end{abstract}

\section{Introduction}

\label{introduction}
Among the very many relevant contributions that Robert Bartnik made to
analysis, geometry and relativity, one of the comparatively less known is his notion of null quasispherical gauge introduced in 1997 \cite{Bartnik97}. Bartnik's main motivation was to develop a framework to write down and study {\em explicit formulations} of  the Einstein vacuum field equations where gauge freedom could be minimized. Bartnik writes ``Ideally, a ‘good’ parametrization of the Einstein equations will have limited or no gauge
freedom (...) with parameters free of constraints and having a direct relation to known radiation parameters''. The null quasi-spherical (NQS) gauge  was an attempt towards this goal. The basic geometric  framework was to consider a foliation of the spacetime by null hypersurfaces and fix the gauge (coordinate) freedom by foliating each null hypersurface by {\em metric spheres}. The null quasi-spherical gauge consisted in expressing the geometry in terms of a foliation defining function $u$, the area radius $r$ of the spheres and standard angular coordinates on each sphere. Thus, the NQS gauge
is a modification of previous characteristic formulations involving other coordinate conditions, such as e.g. Bondi-Sachs coordinates \cite{BS,Sachs} or null
affine coordinates \cite{NU}.  Bartnik's main motivation was to have a framework where numerical implementations of the Einstein field equations could be successfully developed.  And, indeed, later works used the NQS gauge to implement the  vacuum field equations in numerical codes (e.g. \cite{Numer}). However, this was not the only objective, and several applications of the NQS gauge appeared already in \cite{Bartnik97}. In particular, the matching problem across null boundaries of two vacuum spacetimes written in NQS gauge was considered under a number of simplifying assumptions.

While fixing the gauge in one way or another is certainly essential
to do numerical implementations, the completely opposite point of view is  more valuable for other types of problems. It is advantageous to let the gauge remain completely unfixed  in order to have a flexible framework capable of adjusting itself to  different situations. For problems involving hypersurfaces, it is also a very useful point of view to try and separate the geometric properties (both intrinsic and extrinsic)
that have to do with the hypersurface to those that have to do with the ambient space. Of course, the two sets of properties are not unrelated from each other, and this is expressed by equations that link them. This separation is well-understood in the case when the hypersurfaces are spacelike or timelike. Then, the intrinsic geometry is the induced metric, the extrinsic geometry is the second fundamental form and the equations that link them to
ambient properties are the Gauss identity for the connection, as well as the Gauss-Codazzi identities for the curvature. The coordinates in which the hypersurface is described can be chosen completely separated from the coordinates of the ambient space where the hypersurface is embedded. In fact, the separation is so powerful that
one can view the hypersurface as completely detached from the spacetime, and encode the ambient curvature as fields on this detached manifold. This point of view is essential, for instance, when the ambient needs to be reconstructed from the hypersurface, e.g. via the solution of suitable prescribed geometric PDE.

The null case is much harder. Certainly, there is no problem in viewing
the hypersurface as a detached manifold separated from the spacetime and let the embedding map carry all the information about how it sits in the ambient space. However, while in the timelike and spacelike case the hypersurface  carries
a semi-riemannian metric induced by the embedding, which allows one to develop all the necessary tools to do geometry on the abstract space, in the null case  there is no induced metric. So, is there any detached geometry one can construct?

If one does not insist in actually detaching the hypersurface, several approaches can be taken. The first one capable of dealing with general null hypersurfaces was developed by J.A. Schouten \cite{Schouten}. It is based on the use of a
{\em rigging vector}  and it has been used successfully in many contexts. The rigging approach actually works in hypersurfaces of arbitrary signature  \cite{MarsSenovilla93}. A different but related approach is based on the use of {\em screen} distributions
\cite{BejancuDuggal}. A third approach is to construct a riemannian metric on the null hypersurface by combining the induced first fundamental form and the covector metrically associated to the rigging. There are several ways to do this 
\cite{Katsuno, MarsSenovilla93, GutierrezOlea}, the approach in 
\cite{GutierrezOlea} being advantageous in that it requires no extra condition on the rigging.

As already mentioned, it is of interest to detach the hypersurfaces from the ambient space.  In \cite{MarsGRG, Mars2020} a framework has been developed to this purpose. The basic notion is that of {\em  metric hypersurface data}. In a certain sense the construction can be regarded as an abstraction of the rigging construction for general hypersurfaces (and also of the
screen distribution construction for null hypersurfaces). However, there are also differences, the main one being that the metric hypersurface data only encodes  information on the ambient metric along the hypersurface, not of its transverse derivatives, unlike the rigging or screen constructions in the embedded case.
Any metric hypersurface data carries a natural torsion-free connection
$\nablao$. This connection depends only on zeroth order information about the ambient metric on the (abstract) hypersurface. This is an important advantage over other options, since less information is required to
construct it. Besides allowing for a fully detached description, the key idea in
\cite{MarsGRG, Mars2020} was to treat the large freedom inherent to the choice of rigging (at the abstract level) by means of a gauge group acting on geometrically equivalent metric hypersurface data. Although the framework was successful in detaching the hypersurface from the ambient, the definition of null metric
hypersurface data still encodes a priori quantities that, in the embedded picture, correspond to a rigging vector. The main purpose of this paper is to lift this a priori restriction.

We want to start from a abstract manifold $\N$ endowed with a degenerate
symmetric tensor $\gamma$. Besides smoothness,
the only assumption we make on $\gamma$ is that it is minimally degenerate, in the sense that its radical is one-dimensional at each point. We call this
a {\em null manifold}. We emphasize that this restriction is completely natural in the present context, as all null hypersurfaces embedded in a semi-riemannian manifold have this property. It turns out that one can establish a neat relationship
between null manifolds and null metric hypersurface data structures. Specifically, there always exists a covering $\dN$
of $\N$ (which may be one-to-one or two-to-one) that admits a null metric hypersurface data uniquely defined from $(\N,\gamma)$ up to the action of the gauge group. It is remarkable that, even in the 2:1 covering case, the connection $\nablao$ associated to the null metric hypersurface data
on $\dN$ {\em descends to $\N$}. Thus, {\em any} null manifold admits
a collection of torsion-free connections $\{\nablao\}$ related to each other by a suitable gauge group (defined on the covering space). This allows us to define a geometry on any null manifold in a fully covariant and fully detached way. When the covering is one-to-one, the construction recovers the null metric hypersurface data geometry in \cite{MarsGRG, Mars2020}.

Coming back to general relativity, and specifically to the Einstein field equations, one of the main advantages of null hypersurfaces over spacelike or timelike ones is that the equations acquire a hierarchical structure. This has been noted in several contexts, e.g. 
\cite{Sachs2, Rendall, Bartnik97, ChruscielPaetz,  Gabriel1}. In all cases, a foliation of the hypersurface by spacelike sections, as well as suitable
coordinate systems adapted to the foliation are necessary. This fact  restricts the applicability of this results to null hypersurfaces with simple (product) topology. The development of a fully detached and fully covariant
null geometry leads  to the question of whether this hierarchical structure can be identified without the need of assuming a specific foliation and a specific coordinate system. The second aim of this paper is to show that this is indeed possible.

In order to describe the constraint equations of general relativity at the abstract
level it is necessary, first of all, to extend the notion of
null metric hypersurface data by incorporating a symmetric tensor $\bY$ that codifies (abstractly) the extrinsic information of the hypersurface.  This geometric notion is called  hypersurface data \cite{MarsGRG, Mars2020}. In that setup one can define \cite{Gabriel1, ManzanoMars} a symmetric tensor called constraint tensor that, in the null case, codifies at the abstract level
the tangential-tangential components of the ambient Ricci tensor.  This tensor is the equivalent to the null constraint equations of general relativity when written in detached form.

The hierarchical decomposition of the constraint tensor obtained here is based on 
a canonical algebraic decomposition of
arbitrary symmetric $2$-covariant tensors  that we derive from a natural linear map $\tau$ that sends symmetric $(0,2)$-tensors to
symmetric $(2,0)$-tensors. This map is  called ``energy-momentum map'' because of its close connections to the geometry of 
null shells (see \cite{MarsGRG} and Section  \ref{shells}). However, the algebraic decomposition of symmetric tensors in itself is not sufficient to achieve a
hierarchical decomposition of the constraint tensor. The reason is that the constraint tensor involves differential operations on the fields. Thus, to accomplish the
decomposition we need to identify differential operators that respect the
algebraic decomposition. Once this is achieved, the hierarchical decomposition of the constraint tensor can be worked out.

The constraint tensor involves the null metric hypersurface data as well as the extrinsic tensor $\bY$. We carry out the hierarchical decomposition of the constraint tensor only concerning its dependence on the tensor $\bY$. The reason is two-fold. Firstly, it is natural to view the constraints as equations for the extrinsic tensor $\bY$, i.e.
to consider the metric part of the data as given. In this perspective, the terms independent of $\bY$ can be treated as sources, so finding its explicit algebraic decomposition does not affect the structure of the equations. The second reason concerns applying the decomposition  to study the shell equations that arise when two spacetimes are matched across null boundaries. As we shall
see, the terms not involving $\bY$ disappear in this case. This application wraps up the relationship of the present work with Robert Bartnik's paper \cite{Bartnik97}, where these equations were discussed in a specific scenario. Here we find the equations in a fully covariant and detached way, and we write them down in an explicitly hierarchical form without the need
to assume either a special topology for the matching surface, or any specific choice of foliation to do the decomposition. The result is completely general,  so it can be adapted to whatever choices of foliations and of coordinates depending on the specific problem at hand.

The plan of the paper is as follows. In Section \ref{Nullmanifolds} we introduce the notion of {\bf null manifold} and {\bf ruled null manifold} and analyze the relationship  with each other as well as with the notion of null metric hypersurface data. We find that there  exists of a one-to-one correspondence between ruled null manifolds and equivalence classes of null hypersurface metric data
(Proposition \ref{correspondence}) and that any null manifold admits a canonical covering (at most 2:1) which defines a ruled null manifold structure.
In Section \ref{nablaoconnection} we recall the connection
$\nablao$ for null metric hypersurface data and study how it can be extended to general null manifolds. The main result is Theorem \ref{ExistCon} where the existence of a class of connections $\{\nablao\}$ for general null manifolds is established. The rest of the paper is devoted to the decomposition of the constraint tensor and of the shell equations. In Section \ref{EMmap} we study the algebraic properties of the energy-momentum map and obtain as a consequence a canonical decomposition of symmetric $(2,0)$-tensors (Proposition \ref{Decom}). We also study the gauge properties of the
energy-momentum map and their consequences  for the gauge properties of the algebraic decomposition. In Section \ref{EMmap2} we study the PDE consequences of the energy-momentum map and introduce two first order, linear, covariant differential operators that have good properties with respect to the canonical decomposition discussed before. Section \ref{ConstTensorDecom} is devoted  to finding the canonical decomposition of the constraint tensor. As already mentioned we concentrate on the terms that involve the extrinsic part of the data. Finally, in Section \ref{shells} we use the canonical decomposition of the constraint tensor to obtain a completely general, detached and fully covariant
hierarchical decomposition of the shell equations describing, at the purely abstract level, thin concentrations of matter and/or impulsive gravitational waves
propagating on  null hypersurfaces.

\subsection{Notation}

All manifolds are assumed to be second countable, Hausdorff, smooth and without boundary. Recall that such a manifold is automatically paracompact (see e.g. \cite{FWarner}). Unless otherwise stated, all manifolds are also connected.
$\FF(\N)$ denotes the set of smooth real functions on an manifold $\N$ and
$\FF^{\star}(\N)$ the subset of functions that vanish nowhere.
$\XX(\N)$ is the set of smooth vector fields 
and $\XX^{\star}(\N)$ the set of smooth covector fields on $\N$.
A tensor is said to be of type $(p,q)$ when it is $p$-contravariant and $q$-covariant.
The vector space of symmetric $(0,2)$-tensors defined on a vector space $V$ is written as $\Sd(V)$. Similarly $\Su(V)$ stands for the vector space of
symmetric $(2,0)$-tensors on $V$.
The signature of an element of $\Sd(V)$ is written $(p,q,r)$ where $p$ is the number of $+1$, $q$ the number of $-1$ and $r$ the number of $0$ in its canonical form.
As usual, brackets enclosing indices indicate antisymmetrization and parenthesis denote symmetrization. The symmetrized tensor product of two tensors $A$ and $B$ is  $A \otimes_s B := \frac{1}{2} (A \otimes B + B \otimes A)$. We shall use both abstract index notation and index-free notation depending on our needs.
The Lie derivative along a vector field $X$ is $\pounds_{X}$ and the exterior derivative $d$.

\section{Null manifolds}
\label{Nullmanifolds}

Any hypersurface $\N$ embedded in a pseudo-riemannian manifold $(\M,g)$ inherits a
symmetric $(0,2)$-tensor, namely the first fundamental form $\gamma$ defined
as $\gamma:=\phi^{\star}(g)$ if $\phi: \N \longrightarrow \M$ is the embedding. Recall that the
radical of a symmetric $(0,2)$-tensor is defined as
\begin{align*}
  \Rad_{\gamma}|_p := \{ V \in T_p \N, \gamma(V,\cdot) =0\}, \qquad p \in \N.
\end{align*}
The radical is, at every point $p$, a vector subspace of $T_p \N$ and, in the case of hypersurfaces, it is easy to show that its dimension is at most one\footnote{The set of normal vectors
  $\mbox{Nor}|_q$ to $\phi(\N)$ at $q \in \phi(\N)$ is a one-dimensional vector space. The push-forward of $\Rad_{\gamma}|_p$  lies in $\mbox{Nor}|_{\phi(p)}$. Since
  $\phi_{\star}$ is injective the dimension of $\Rad_{\gamma}|_p$ is the same as the dimension of its image, so the bound follows. See also \cite{Mars2020}[Lemma 2.2] where the result is proved in the more general context of metric hypersurface data.}
\begin{align*}
  \dim (\Rad_{\gamma}|_p) \leq 1.
\end{align*}
A non-trivial radical means that the first fundamental form is not a metric or, equivalently, that at each point $p \in \phi(\N)$ there is a non-zero normal vector which is at the same time tangential. Such hypersurfaces are called {\em null} and obviously they can only exist when the metric $g$ is not definite (namely, its signature is neither $(0,\n+1,0)$ nor $(\n+1,0,0)$ where $\n+1$ is the dimension of $\N$). Any attempt
to study null hypersurfaces from a purely detached point of view requires dealing with manifolds endowed with a degenerate symmetric $(0,2)$-tensor $\gamma$
whose radical is one-dimensional at every point. Thus, we put forward the following definition.
\begin{definition}
  A {\bf null manifold} $(\N,\gamma)$ is a manifold of dimension $\n \geq 1$ endowed with a smooth, symmetric $(0,2)$-tensor $\gamma$ satisfying
  \begin{align*}
    \dim (\Rad|_{\gamma}|_p) = 1 \qquad \forall p \in \N.
  \end{align*}
\end{definition}
One of our aims in this paper is to define a geometry on  null manifolds. The main complication of course arises from the fact that $\gamma$ is not a metric, so  in particular there is no Levi-Civita connection associated to $\gamma$.

We start by analyzing some geometric consequences of the definition.
The (disjoint) union of all vector spaces $E= \bigcup_{p \in \N} \Rad_{\gamma}|_p$ defines a vector bundle over $\N$. We call $\pi: E \rightarrow \N$ the projection map. This bundle is a (real) line bundle.
It is a standard fact (see e.g. \cite{Milnor}[Theorem 2.2]) that a line bundle is trivial (i.e. isomorphic
to the product bundle $\pi_1 : \N \times \mathbb{R} \rightarrow \N$)
if and only if it admits a nowhere zero section, i.e. a smooth map $\upn : \N \rightarrow E$ satisfying
$\pi \circ \upn = \Id_{\N}$ such that $\upn|_p \neq 0$ for all $p \in \N$.

Any real line bundle  $\pi: E \longrightarrow  \N$
admits a trivial double covering. More precisely there exists a manifold $\dN$ and a double cover of the base 
\begin{align}
f  : \dN \longrightarrow \N \label{covering}
\end{align}
such that the line bundle $f^{\star} (E)$ is trivial. When the original bundle
$(E,\N,\pi)$ is already trivial then $\dN$ is just 
$\N \times \{-1,1\}$ , so it consists simply of two copies of $\N$.

Line bundles are classified by its first Stiefel-Whitney characteristic class.
The first Stiefel-Whitney class of a line bundle $(E,\N,\pi)$ is denoted by
$\omega_1(E)$ and is an element of
$H^1 (\N, \mathbb{Z}_2)$, the first cohomology of $\N$ with coefficients in $\mathbb{Z}_2  := \mathbb{Z} \mbox{ mod } 2$.
A line bundle is trivial if and only if  $\omega_1(E) = 0$. In particular, when $H^1 (\N, \mathbb{Z}_2)$ contains only the zero element, all line bundles over
$\N$ are necessarily trivial. Of course, this happens in particular when $\N$ is simply connected or contractible.

The line bundle of a null manifold is not just any line bundle. It is
associated to a one-dimensional distribution on $\N$, which in turn is associated to the tensor $\gamma$. Note that $\gamma$ having a one-dimensional radical at every point  is equivalent to saying that $\gamma$ has signature $(p, q, 1)$ with $p+q = \n-1$.
The numbers $q$ and $p$ are locally constant, so they are constant if
$\N$ is connected.

Every manifold is known to admit a (positive definite) riemannian metric
(see e.g. \cite{Lee}[Proposition 13.3]). It is natural to ask whether there  are any obstructions to define a null manifold structure on a manifold $\N$. In the case of general signature, the only result we know is a theorem by Bel{}\'{}ko \cite{Belko} which, particularized to the present setup, states that $\N$ admits a null manifold structure $(\N,\gamma)$ of signature
$(p,q,1)$ if and only if $\N$ admits three mutually complementary
smooth distributions of dimensions $p$, $q$ and $1$ respectively. 

From a physical point of view, the most interesting case is when $\gamma$ is positive semidefinite, as this is the situation that arises in null hypersurfaces embedded in Lorentzian manifolds. As we show in the next lemma, the only condition for $\N$ to admit a semi-definite null manifold structure is that
$\N$ admits a one dimensional distribution, i.e.
a line bundle $(E, \N ,\pi)$ with $\pi^{-1} (p) \subset
T_p \N$, for all $p \in \N$.
\begin{lemma}
  \label{distribution}
  An $\n$ dimensional manifold $\N$ admits a null manifold structure
  $(\N,\gamma)$ of signature $(\n-1,0,1)$ (or $(0,\n-1,1)$) if and only if
  $\N$ admits a one-dimensional distribution. 
\end{lemma}
\begin{proof}
Necessity is obvious. For sufficiency take a riemannian metric $g$ in $\N$. At every point
$p \in \N$ consider the pair of vectors $\{ \pm e \}  \in \pi^{-1} (p)$
  defined by
  $g|_p (e,e) = 1$. Although in general there is no way to extract a vector field on $\N$ from this collection of vectors (indeed, this occurs if and only if
  the bundle is trivial), the tensor field
  \begin{align*}
    \gamma := g - g(e,\cdot) \otimes g(e, \cdot)
  \end{align*}
  is nevertheless well-defined because its expression is invariant under change of sign in $e$. It is clear that  $\gamma$ has signature $(\n-1,0,1)$ so the claim follows.
    \end{proof}

    So, in positive (or negative) semi-definite signature
    the existence of a null manifold structure boils down
    to the existence of a one-dimensional distribution. This has the following corollary.
    \begin{corollary}
      \label{Euler}
      Let $\N$ be an $\n$-dimensional manifold. If $\N$ is not compact then it always admits a positive (or negative) semi-definite
      null manifold structure. If $\N$ is compact, it admits a
positive/negative semi-definite
null manifold structure if and only if its Euler characteristic is zero.
\end{corollary}
\begin{proof}
  It is a standard fact that any non-compact manifold admits a nowhere zero vector field. Thus, it also admits a one-dimensional distribution.  For the compact case, existence of a one-dimensional distribution is equivalent to
  the Euler characteristic $\chi(\N)$ of $\N$ being zero \cite{Markus}[Theorem 3].
\end{proof}

It is a well-known fact that compact manifolds have vanishing Euler characteristic if and only if they admit a nowhere zero vector field. Indeed, if $\N$ admits a nowhere zero vector field, then a direct application of the
Poincar\'e-Hopf index theorem implies $\chi(\N)=0$. The converse was proved by Hopf in \cite{Hopf}. It is also well-known (see e.g. \cite{Oneill}) that a manifold admits a Lorentzian metric if and only if it admits a nowhere zero vector field. So, the previous corollary can also stated as
\begin{corollary}
  \label{nowherezero}
  An $\n$-dimensional manifold $\N$  admits a positive (or negative) semi-definite
  null manifold structure if and only if it admits a nowhere zero vector field. Equivalently, if and only if $\N$ admits a Lorentzian metric.
\end{corollary}
Note that the vector field in this corollary has in general no relation with the one-dimensional distribution in Lemma \ref{distribution}.

As we shall see below the definition of a geometry on $(\N,\gamma)$ is simplest when this line bundle is trivial. Moreover, this case will serve as the basis for the more general setup, so we introduce the following definition.
\begin{definition}
  A {\bf ruled null manifold} is a null manifold $(\N,\gamma)$
  for which the line bundle 
  $(E:= \bigcup_{p \in \N} \Rad_{\gamma}|_p$, $\N ,\pi)$ is trivial.
\end{definition}
Understanding ruled null manifolds is relevant for the general null manifold case thanks to the covering map \eqref{covering}. Indeed, consider a null manifold which is not ruled, i.e. such that $(E:= \bigcup_{p \in \N} \Rad_{\gamma}|_p, \N ,\pi)$ is non-trivial. As already discussed, the pull-back bundle $f^{\star} (E)$  over $\dN$
is a trivial line bundle. Moreover, it is a subbundle of the tangent bundle of $\dN$ because from $f$ being a local diffemorphism it follows
$f^{\star} (T \N) \approx T \dN$, where the diffeomorphism is in the sense of vector bundles. Since $E$ is by construction a subbundle of
$T\N$, the pull-back $f^{\star} (E)$ is, via the isomorphism above, also
a subbundle of  $T\dN$. Hence, it defines a one-dimensional distribution in $\dN$. The line bundle $f^{\star} (E)$ being trivial, it admits a nowhere vanishing section. As already said, this section can be identified with a nowhere zero vector field $\dupn \in \XX(\dN)$.
Now $\dgamma := f^{\star} (\gamma)$ is a symmetric two covariant vector field on $\dN$. Its radical is one-dimensional at every point, and by construction
$\dupn$ spans this radical. Consequently the null manifold structure
$(\N,\gamma)$ induces a ruled null manifold structure $(\dN, \dgamma)$ on $\dN$. Thus, understanding ruled manifold structures is relevant not only for its own sake but, particularly, for the consequences  it has concerning the  more general case of null manifold structures.

We shall study ruled manifolds via the concept of {\bf metric hypersurface data}. This notion was introduced in \cite{MarsGRG, Mars2020} with a precursor appearing already in
\cite{MarsSenovilla93}. Further properties of this notion can be found in
\cite{ManzanoMars}
and applications to the matching problem have been discussed  in \cite{Miguel1, Miguel2}
while applications to characteristic initial value problem appear in \cite{Gabriel1, Gabriel2}. For the purposes of this paper we shall restrict to a particular case, namely the so-called {\bf null metric hypersurface data}. The definition is as follows \cite{Mars2020}.
\begin{definition}
  {\bf Null metric hypersurface} data is a four-tuple $\metdata$ where $\N$ is a smooth manifold of dimension $\n \geq 1$, $\gamma$ is a smooth symmetric $(0,2)$-tensor field with one-dimensional radical at every point, $\ellc \in \XX^{\star}(\N)$ and $\elltwo \in \FF(\N)$, provided the following symmetric $(0,2)$-tensor $\A|_p$ defined on $T_p \N \oplus  \mathbb{R}$
  \begin{align}
    \A|_p ( (X_1,a_1),(X_2,a_2)) := \gamma(X_1,X_2) + a_1 \ellc(X_2) + a_2 \ellc(X_1) + a_1 a_2 \elltwo \label{defA}
     \end{align}
    is non-degenerate at every $p \in \N$.
\end{definition}
Since the tensor $\A|_p$ is non-degenerate, there is a unique symmetric contravariant tensor $\A^{\sharp}|_p$ on $T_p \N \oplus \mathbb{R}$ defined by
\begin{align*}
  \A^{\sharp}|_p  ( \bm{(X_1,a_1)} , \cdot )  = (X_1,a_1), \qquad \quad \forall (X_1,a_1)
  \in T_p \N \oplus \mathbb{R}
  \end{align*}
where the covector $\bm{(X_1,a_1)}$ is associated to $(X_1,a_1)$ by
$\bm{(X_1,a_1)} (\cdot) := \A|_p ((X_1,a_1), \cdot)$. One can then define
  the tensors $\{ P|_p, \upn|_p\}$ on $T_p \N$ by splitting $\A^{\sharp}$, namely
  \begin{align*}
    P (\bm{\omega_1}, \bm{\omega_2}) = \A^{\sharp}|_p ((\bm{\omega_1},0),(\bm{\omega_2},0)),
    \qquad \upn|_p(\bm{\omega_1}):= \A^{\sharp}|(\bm{\omega_1},0), (0,1)),
    \qquad \bm{\omega_1}, \bm{\omega_2} \in T^{\star}_p \N.
  \end{align*}
  $P$ defines a $(2,0)$-tensor field and $\upn$ a vector field on $\N$.
  In abstract index notation, they are defined by \cite{MarsGRG, Mars2020}
\begin{align}
\gamma_{ab} \vecn^b  & = 0, \label{prod1} \\
\ell_a \vecn^a  & = 1, \label{prod2}  \\
P^{ab} \ell_b + \elltwo \vecn^a & = 0,  \label{prod3} \\
P^{ab} \gamma_{bc} + \vecn^a \ell_c & = \delta^a_c. \label{prod4}
\end{align}
Note that by construction $\upn$ is a nowhere zero section of the
bundle $(E = \bigcup_{p \in \N} \Rad_{\gamma}|_p,\N,\pi)$. So, the null manifold $(\N,\gamma)$ associated to any
null metric hypersurface data is always a ruled null manifold.  It is natural to ask whether the converse is also true, namely whether one can associate
null metric hypersurface data to any ruled null manifold, and if so in how many fundamentally different ways  this is possible.

To answer this question we first need to recall the fact that
null metric hypersurface data has a natural built-in gauge group and that two null metric hypersurface data related by a gauge transformation are to be regarded as equivalent from a geometric point of view. We first recall the action of the gauge group (details can be found in \cite{Mars2020}).

Define the following internal operation on the set $\G :=\FF^{\star}(\N) \times \XX(\N)$
\begin{align*}
  \cdot : \G \times \G & \longrightarrow \G \\
  ( (z_2,\gauge_2) , (z_1, \gauge_1) ) & \mapsto
  (z_2 z_1, \gauge_1 + z_1^{-1} \gauge_2).
\end{align*}
This operation endows $\G$ with a group structure. This group is called {\em hypersurface data gauge group} or simply {\em gauge group} if no confusion can arise. The inverse of $(z,\gauge)$ and the neutral element $e$ are
\begin{align*}
  (z,\gauge)^{-1} = (z^{-1}, - z \gauge), \qquad e = (1,0).
\end{align*}
Given $(z,\gauge) \in \G$, its action 
on the null metric hypersurface data, denoted by $\G_{(z,\gauge)}$, is defined as follows
\begin{align}
  \G_{(z,\gauge)} (\gamma)= \gamma, \qquad
  \G_{(z,\gauge)} (\ellc) = z ( \ellc + \gamma(\gauge, \cdot)), \qquad
  \G_{(z,\gauge)} (\elltwo) = z^2 \big ( \elltwo + 2 \ellc(\gauge) +
  \gamma(\gauge,\gauge) \big ). \label{gaugemhd}
\end{align}
One checks easily that the action is well-defined, i.e. takes
metric hypersurface data ${\mathcal D}:= \{\N,\gamma,\ellc,\elltwo\}$ and produces
metric hypersurface data
$\G_{(z,\gauge)} ({\mathcal D}):= \{\N,\G_{z,\gauge}(\gamma),\G_{(z.\gauge)}(\ellc),\G_{z,\gauge)}(\elltwo)\}$,
and one also checks that this action is a realization of the group, namely
\begin{align*}
\G_{(z_1, \gauge_1)} \circ \G_{(z_2,\gauge_2)} =
  \G_{(z_1, \gauge_1) \cdot (z_2, \gauge_2)}.
  \end{align*}
The gauge behaviour of the contravariant tensors $\{ P,\upn\}$ is obtained from \eqref{prod1}-\eqref{prod4}. The result is
\cite{MarsGRG}
\begin{align}
  \G_{(z,\gauge)} (P) = P - 2 \upn \otimes_s \gauge,
  \qquad \quad
  \G_{(z,\gauge)} (n) = z^{-1} \upn. \label{Pprime}
\end{align}
This gauge transformation leads to a notion of
geometric equivalence of metric hypersurface data.
\begin{definition}
  Let $\D_1:= \{\N,\gamma_1, \ellc_1, \elltwo_1\}$ and
  $\D_2 := \{\N,\gamma_2, \ellc_2, \elltwo_2\}$ two null metric hypersurface data on the same manifold
  $\N$. We say that they are {\bf geometrically equivalent} and write
  \begin{align*}
    (\N,\gamma_1, \ellc_1, \elltwo_1) \sim   (\N,\gamma_2, \ellc_2, \elltwo_2)
  \end{align*}
  whenever there exists  a group element $(z,\gauge) \in \G$ such that
  $\G_{(z,\gauge)} (\D_1) = \D_2$.
  \end{definition}
It is obvious that $\sim$ defines an equivalence relation on the set null metric hypersurface data on a given manifold $\N$. We shall denote the equivalence class of $\metdata$ with a bar, namely $\overline{\D}$ or 
$\metdatabar$

We have already stated that any metric hypersurface data defines a ruled null geometry
$(\N,\gamma)$. By  the gauge invariance of $\gamma$ it is clear that this null geometry only depends on the equivalence class. We now want to show that there exists a one-to-one correspondence between ruled null geometries and equivalence classes of metric hypersurface data.
We start by finding  under which conditions a given covector and scalar fields defined on a ruled null manifold define null metric hypersurface data.
\begin{lemma}
  \label{conditionsmetdata}
  Let $(\N,\gamma)$ be a ruled null manifold. Let $\ellc \in \XX^{\star}(\N)$
  and $\elltwo \in \FF(\N)$ be given. Then $\metdata$ is  null metric hypersurface data if and only if for one (and hence any) nowhere zero section  $e_1$ of
  $(E,\N,\pi)$ it holds $\ellc (e_1) \neq 0$ everywhere
  (in other words,  $\ellc$ is transverse to $\Rad_{\gamma}$ at every point).
  \end{lemma}
  \begin{proof}
   If $\metdata$ is null metric hypersurface data,
    we may chose $e_1|_p=\upn|_p$ and the condition $\ellc(e_1)\neq 0$ is satisfied.
    To prove the converse we only need to make sure that the symmetric $(0,2)$-tensor $\A|_p$ on $T_p \N \oplus \mathbb{R}$ defined in \eqref{defA}
    is non-degenerate. Let $e_1$ be a non-zero section of $(E,\N,\pi)$. By assumption $\ellc|_p (e_1) \neq 0$ everywhere. We work at a fixed point $p \in \N$ from now on.  Define $\n := \dim(\N)$ and complete $e_1$ to a  canonical basis $\{ e_a \}$ of $\gamma|_p$, i.e. a basis satisfying
    \begin{align*}
      \gamma(e_a, e_b) =0 \quad \mbox{if} \quad a \neq b, \qquad
      \gamma(e_a,e_a) = \epsilon_a
    \end{align*}
    with $\epsilon_1=0$ and $\epsilon^2_A =1$ if $2 \leq A \leq \n$. Then, 
    the vectors
\begin{equation*}
  E_0:=(V,1),\qquad E_a:=(e_a,0),\qquad\textup{with}\qquad V:=-\sum_{B=2}^{\n}\epsilon_B \ellc\vert_p(e_B)e_B, \qquad V \in T_p\mathcal{N}
\end{equation*}
constitute a basis of $T_p\mathcal{N}\times\mathbb{R}$. From \eqref{defA} we get 
\begin{align*}
\A\vert_p(E_0,E_0)&=\gamma\vert_p\lp V,V\rp+2\ellc\vert_p\lp V\rp+ \ell^{(2)}, & \A\vert_p(E_0,E_1)&=\ellc\vert_p\lp e_1\rp,\\
\A\vert_p(E_0,E_A)&=\gamma \vert_p\lp V,e_A\rp+\ellc\vert_p\lp e_A\rp=0, & \A\vert_p(E_1,E_1)&=0,\\
\A\vert_p(E_1,E_A)&=0, & \A\vert_p(E_A,E_B)&=\gamma \vert_p\lp e_A,e_B\rp=\delta_{AB}\epsilon_A.
\end{align*}
The determinant of $\A|_p$ in this basis is
\begin{align*}
  \mbox{det} (\A|_p)= - (\ellc|_p (e_1))^2 \prod_{B=2}^\n \epsilon_B,
\end{align*}
which is non-zero. Hence, $\A|_p$ is non-degenerate, as claimed.
  \end{proof}

  The following lemma, first proved in \cite{ManzanoMars}, is key to establish the relationship between ruled null manifolds and equivalent classes of null hypersurface data. We include the proof in order to make this paper as self-contained as possible. 
  \begin{lemma}
    \label{gaugefix}
    Let $\metdata$ be null metric hypersurface data.
    Let $\bm{\omega} \in \XX^{\star}(\N)$ and $u \in \FF(\N)$. Assume that $\bm{\omega} ( \upn) \neq 0$ everywhere. Then
  there exists a unique gauge transformation $\G_{(z,\gauge)}$ satisfying
  \begin{align}
    \G_{(z,\gauge)} (\ellc) = \bm{\omega}, \qquad \G_{(z,\gauge)} (\elltwo) = u.
    \label{trans}
  \end{align}
  Moreover, the gauge group element $(z,\gauge)$ is given by
  \begin{align}
    z = \bm{\omega} (\upn), \qquad \gauge = \frac{1}{\wn} P(\bm{\omega}, \cdot)
      + \frac{u - P(\bm{\omega},\bm{\omega} )}{2 \left ( \bm{\omega}(\upn) \right )^2}      \upn.
      \label{gaugeparameter}
  \end{align}
\end{lemma}

\begin{remark}
  The condition $\bm{\omega}(n) \neq 0$ is necessary by Lemma \ref{conditionsmetdata}.
 \end{remark}

\begin{proof}
  We  first assume that the gauge transformation exists and use this fact to
  restrict its form up to a free function. We then restrict ourselves to such  class of group elements
  and show that there exists precisely one group element satisfying \eqref{trans},  and that this is given by  \eqref{gaugeparameter}. This will prove both the existence and uniqueness claims of the lemma. For the first part we impose \eqref{trans} namely
      \begin{align}
    \G_{(z,\gauge)} (\ellc) = z \left ( \ellc + \gamma(\gauge,\cdot) \right ) = \bm{\omega} , \qquad
    \G_{(z,\gauge)} (\elltwo) = z^2 \big ( \elltwo + 2 \ellc(\gauge) + \gamma(\gauge,\gauge) \big ) = u.
    \label{trans2}
  \end{align}
  Contracting the first with $\upn$ gives $z = \bm{\omega}(\upn)$, so
  \begin{align*}
    \bm{\rho}:= \gamma(\gauge,\cdot )  = \frac{1}{\wn} \bm{\omega} - \ellc.
  \end{align*}
  Note that $\bm{\rho}(\upn)=0$. The computation $
    \gamma_{ab} \left ( \gauge^b - P^{bc} \rho_c \right )
=    \rho_{a} - \left ( \delta^c_a - n^c \ell_a \right ) \rho_c =0$
shows that the 
 vector $\gauge - P(\bm{\rho}, \cdot)$ lies in the kernel of $\gamma$, so there exists a function $f$ such that
    \begin{align*}
      \gauge^a  = P^{ab} \rho_b + f n^b = \frac{1}{\wn} P^{ab} \omega_b
      + \left ( \elltwo  + f \right ) n^a.
    \end{align*}
    Thus, it suffices to restrict oneself to gauge parameters in the class
    \begin{align}
   \left \{ \Big ( z = \wn , \gauge  = \frac{1}{\wn} P(\bm{\omega}, \cdot)  + q \upn\Big  ) ,  \quad q \in \FF(\N) \right  \}.
       \label{defV}
    \end{align}
    We now start anew and prove that there is precisely one function $q$ such that the corresponding $(z,\gauge)$ in  \eqref{defV}
    fulfills 
    conditions \eqref{trans}. For $\gauge$ as  in \eqref{defV}  we get
    \begin{align}
       \gamma( \gauge, \cdot )  & = \frac{1}{\wn} \gamma( P (\bm{\omega}, \cdot),
      \cdot ) =  \frac{1}{\wn} \bm{\omega} - \ellc,  \label{V1}  \\
          \bm{\omega} (\gauge)  & = \frac{1}{\wn} P(\bm{\omega}, \bm{\omega} )
          + q \wn,  \nonumber \\
                \ellc (\gauge)  & = - \elltwo + q, \label{V2} \\
              \gamma(\gauge,\gauge)  & = \frac{1}{\wn} \bm{\omega}(\gauge) - \ellc(\gauge)  =
      \frac{P(\bm{\omega}, \bm{\omega})}{\wn^2} + \elltwo. \label{V3}
    \end{align}
    The first condition in \eqref{trans2} is satisfied for all $q$ as a direct consequence of \eqref{V1}. From \eqref{V2} and \eqref{V3}, the
    second in \eqref{trans2} is satisfied if and only if
    \begin{align*}
      \wn^2 \left ( 2 q + \frac{P(\bm{\omega}, \bm{\omega})}{\wn^2}
      \right ) = u \qquad \Longleftrightarrow \qquad
      q =  \frac{u - P(\bm{\omega}, \bm{\omega})}{2 \wn^2}.
    \end{align*}
      \end{proof}
      Lemma \ref{gaugefix} has the following interesting consequences.
  \begin{corollary}
    \label{equivdata}
    Let $\N$ be a smooth manifold and  $\D_1:= \metdataone$, $\D_2:=\metdatatwo$ two null metric
    hypersurface data defined on $\N$. Then, there is a gauge group element $(z,\gauge) \in \FF^{\star}(N) \times \Gamma(T^{\star} \N)$ such that
    $\G_{(z,\gauge)} ( \D_1 ) = \D_2$ if and only if $\gamma_1= \gamma_2$.
  \end{corollary}
  
  \begin{proof}
    The necessity is obvious because $\gamma$ remains unchanged under a gauge transformation. Sufficiency is a direct application of Lemma \ref{gaugefix} to
     $\bm{\omega}= \ellc{}_2$ and $u = \elltwo_2$. 
  \end{proof}

  \begin{proposition}
    \label{correspondence}
    Let $\N$ be a smooth manifold. There is a one-to-one correspondence between ruled null manifold structures on $\N$ and equivalence classes of null metric hypersurface data
    in $\N$.
  \end{proposition}

  \begin{proof}
    We already know that an equivalence class of null metric hypersurface data
    $\metdatabar$ defines a ruled null manifold structure $(\N,\gamma)$. To prove the converse we consider a ruled null manifold structure $(\N,\gamma)$ on $\N$ (note this may not exist, in which case there are no null metric hypersurface structures either, and there is nothing to prove). We need to show two things:
    (i) that one can complete $(\N,\gamma)$ to a null metric hypersurface data $\metdata$ and (ii) that any two null metric hypersurface data of this form belong to the same equivalence class. The validity of (ii) is the content of Corollary \ref{equivdata}, so it only remains to establish (i). By definition of ruled null manifold, there exists a smooth nowhere zero vector field $\upn$ such that
    $\upn|_p \in \Rad_{\gamma}|_p$ for all $p \in \N$. We have already mentioned
    the fundamental fact that any smooth manifold can be endowed with a smooth riemannian metric.    Let $g_0$ one such metric and
    define
    \begin{align*}
      \ellc := \frac{1}{g_0(\upn,\upn)} g_0(\upn, \cdot)
    \end{align*}
    By construction $\ellc$ is a smooth covector field that satisfies $\ellc(\upn) =1$. Thus, by lemma \ref{conditionsmetdata} the data $\metdata$ where $\elltwo$ is any smooth real function (e.g. $\elltwo=0$) defines null metric hypersurface data. 
      \end{proof}

      The following corollary determines the
      necessary and sufficient conditions for a manifold $\N$
      to admit
      null metric hypersurface data $\metdata$ with positive semidefinite
      $\gamma$.
      \begin{corollary}
  \label{Euler2}
  Let $\N$  be a smooth manifold of dimension $\n$. Then $\N$ can be endowed with  null
  metric hypersurface data $\metdata$ with $\gamma$ of signature
  $(\n-1,0,1)$ (equivalently, with
  a ruled null manifold structure $(\N,\gamma)$) if and only if $\N$ admits a nowhere zero vector field, i.e. if and only if either $\N$ is non-compact, or it is compact with vanishing Euler characteristic.
\end{corollary}

\begin{proof}
  Ruled manifold structures are in particular null manifold structures, so existence of a nowhere-zero vector field is necessary by Corollary \ref{nowherezero}. For sufficiency, simply pick up a riemannian metric $g_0$ on $\N$ and a nowhere zero vector field $X \in \XX(\N)$ and define
  \begin{align*}
    \gamma := g_0 - \frac{1}{g_0(X,X)} g_0 (X ,\cdot) \otimes g_0(X,\cdot)
    \end{align*}
    This tensor has signature $(\n-1,0,0)$ and $\Rad_{\gamma} = \mbox{span} (X)$.\end{proof}
 
      \begin{remark}
        Since all compact manifolds of odd dimension have vanishing Euler characteristic (e.g. \cite[Corollary 3.37]{Hatcher}), it follows from 
Corollaries \ref{Euler} and \ref{Euler2}
        that every  manifold $\N$ of odd dimension admits a null manifold structure and also a (different in general) ruled manifold structure (or equivalently  null metric hypersurface data). This applies in particular to the physical dimension of four-dimensional ambient spaces. 
      \end{remark}

      \section{The $\nablao$ connection}
\label{nablaoconnection}
      A key property of null hypersurface data is that it admits a well-defined covariant derivative with good gauge behaviour. In this section we recall this result and show in what sense it can be extended to  general  null manifolds.

      Given a null metric hypersurface data $\metdata$, it is convenient to introduce the following $(0,2)$-tensor fields
      \begin{align}
\bU  & \defi  \frac{1}{2} \pounds_{\upn} \gamma \label{defU} \\
\bF & \defi \frac{1}{2} d \ellc, \label{defF}
\end{align}
as well as the covector $\bsone$ 
\begin{align*}
\bsone =  \iota_{\upn} \bF.
\end{align*}
where $\iota$ means contraction in the first index, i.e.
$\sone_b = \F_{ab} n^a$. It is easy to show that  \cite{Gabriel1}
\begin{align*}
\bU (n, \cdot ) = 0, \qquad \quad
d \bsone  = \pounds_{\upn} \bF. 
\end{align*}
The following result, proved in \cite{Mars2020}, establishes the existence
of the torsion-free connection $\nablao$ as well as its gauge behaviour. We use the standard definition of the difference $(\nabla^{(1)}- \nabla^{(2)})$ 
between two connections $\nabla^{(1)}$ and
$\nabla^{(2)}$, namely the $(1,2)$-tensor given by
\begin{align*}
  (\nabla^{(1)}- \nabla^{(2)}) (X,Z) = \nabla^{(1)}_X Z - \nabla^{(2)}_X Z.
\end{align*}
\begin{proposition}
\label{ExistConn}
Let $\{ \N,\gamma,\ellc,\ll\}$ be null metric hypersurface data. 
There exists a unique torsion-free connection $\nablao$, called
{\bf metric hypersurface connection} defined by the two properties
\begin{align}
(\nablao_X \gamma) (Z,W) & =
- \bU(X,Z)\ellc(W) 
-  \bU(X,W)\ellc(Z) 
  \label{Cond1} \\
  (\nablao_X \ellc) (Z) + (\nablao_Z \ellc) (X) & = - 2 \ll \bU(X,Z), \qquad
       X,Z,W \in \XX(\N).                                           
\label{Cond2}
\end{align}
Moreover, under a gauge transformation with gauge parameters $(z,\gauge)$, the connection transforms as 
\begin{align}
  \G_{(z,\gauge)} \nablao = \nablao
     + \gauge \otimes \bU  + \frac{1}{2z } \upn \otimes \left ( \pounds_{z \gauge} \gamma
    + 2 \ellc \otimes_s dz \right ). \label{gaugenablao}
      \end{align}
\end{proposition}
The connection $\nablao$ has good behaviour under the gauge group, in the sense that the transformation \eqref{gaugenablao} realizes the group structure.
\begin{lemma}
  Let $\{ \N,\gamma,\ellc,\ll\}$ be null metric hypersurface data. The connection
  $\nablao$ satisfies
  \begin{align*}
    \left ( \G_{(z_2,\gauge_2)} \circ \G_{(z_1,\gauge_1)} \right ) (\nablao)
    = \G_{(z_2,\gauge_2) \cdot (z_1,\gauge_1)} (\nablao)
  \end{align*}
  for any pair of group elements $(z_1,\gauge_1), (z_2, \gauge_2) \in \G$.
\end{lemma}

\begin{proof}
We shall need the following easy fact.  For any vector field $X$ and scalar function $f$ it holds
\begin{align}
\pounds_{f X} \gamma = f \pounds_{X} \gamma + 2 df \otimes_s \uwidehat{X}.
\label{idenLiegamma}
\end{align}
where for any vector $X \in \X(\N)$ we  define $\uwidehat{X} := \gamma(X,\cdot)$.
Given that $\upn$ lies in the radical of $\gamma$, the gauge behaviour of 
$\bU$ follows directly from \eqref{idenLiegamma}
  \begin{align*}
    \G_{(z,\gauge)} (\bU) = \frac{1}{2} \pounds_{z^{-1} \upn} \gamma
    = \frac{1}{2z} \pounds_{\upn} \gamma = z^{-1} \bU.
  \end{align*}
  Expanding the Lie derivative term in \eqref{gaugenablao} the gauge transformed connection can also be written as
\begin{align}
  \G_{(z,\gauge)}( \nablao) & = \nablao
     + \gauge \otimes \bU  + \frac{1}{2} \upn \otimes  \pounds_{\gauge} \gamma
  +  \upn \otimes \left ( \ellc + \uwidehat{\gauge} \right ) \otimes_s \frac{dz}{z}  \label{gaugenabla1}\\
                          & = \nablao + z \gauge \otimes \G_{(z,\gauge)} \bU
                            + \frac{1}{2} \upn \otimes \pounds_{\gauge} \gamma + \G_{(z,\gauge)}(\upn) \otimes \G_{(z,\gauge)}
                            (\ellc) \otimes_s \frac{dz}{z}
\label{gaugenabla}
\end{align}
Now we apply this expression to $\nablao \rightarrow
\G_{(z_1,\gauge_1)}(\nablao)$
and $(z,\gauge) \rightarrow (z_2, \gauge_2)$. Since the
connection $\G_{(z_1,\gauge_1)}(\nablao)$
is simply $\nablao$ in the transformed gauge, the expression above needs to be written with all metric hypersurface data quantities expressed in the new gauge. Thus,
\begin{align}
  \left (  \G_{(z_2,\gauge_2)} \circ \G_{(z_1,\gauge_1)} \right ) (\nablao)  = &  \,\, \G_{(z_1,\gauge_1)} (\nablao) + z_2 \gauge_2 \otimes 
\left (\G_{(z_2,\gauge_2)} \circ \G_{(z_1,\gauge_1)} \right ) (\bU)  
  \label{productG} \\
& + \frac{1}{2} \G_{(z_1,\gauge_1)} (\upn) \otimes
\pounds_{\gauge_2} \gamma 
+ \left ( \G_{(z_2,\gauge_2)} \circ \G_{(z_1,\gauge_1)}  \right  ) (\upn) 
\otimes \left ( \G_{(z_2,\gauge_2)} \circ \G_{(z_1,\gauge_1)} \right ) (\ellc) 
\otimes_s \frac{dz_2}{z_2}. \nonumber
\end{align}
We now write $(z_2, \gauge_2) \cdot (z_1, \gauge_1)$ as
$(z_3,\gauge_3)$ and apply \eqref{gaugenabla} with $(z,\gauge) \rightarrow
(z_1, \gauge_1)$ to replace the first term in the right-hand side of \eqref{productG}. This yields, after using that $z_3 = z_2 z_1$,
\begin{align*}
  \left (  \G_{(z_2,\gauge_2)} \circ \G_{(z_1,\gauge_1)} \right ) (\nablao)  = &  
\nablao + \left ( \gauge_1 + \frac{1}{z_1} \gauge_2 \right ) \otimes \bU 
+ \frac{1}{2} \upn \otimes \left ( \pounds_{\gauge_1} \gamma + 
\frac{1}{z_1} \pounds_{\gauge_2} \gamma 
+ 2 (\ellc + \uwidehat{\gauge}_1 ) \otimes_s \frac{dz_1}{z_1} 
\right ) \\
& + \G_{(z_3,\gauge_3)}   (\upn) 
\otimes \G_{(z_3,\gauge_3)}  (\ellc) 
\otimes_s \frac{dz_2}{z_2}.  
\end{align*}
We now apply \eqref{idenLiegamma} and get, recalling also that $\gauge_3 = 
\gauge_1 + z_1^{-1} \gauge_2$,
\begin{align*}
  \left (  \G_{(z_2,\gauge_2)} \circ \G_{(z_1,\gauge_1)} \right ) (\nablao)  = &  
\nablao + \gauge_3  \otimes \bU 
+ \frac{1}{2} \upn \otimes \left ( 
\pounds_{\gauge_3} \gamma + 2 \uwidehat{\gauge}_2 \otimes_s \frac{dz_1}{z^2_1} 
 + 2 (\ellc + \uwidehat{\gauge}_1 ) \otimes_s \frac{dz_1}{z_1} 
\right ) \\
& + \G_{(z_3,\gauge_3)}   (\upn) 
\otimes \G_{(z_3,\gauge_3)}  (\ellc) 
\otimes_s \frac{dz_2}{z_2}   \\
= & \nablao + \gauge_3  \otimes \bU 
 + \frac{1}{2} \upn \otimes \pounds_{\gauge_3} \gamma +
\upn \otimes \left ( \ellc + \uwidehat{\gauge}_1 + z_1^{-1} \uwidehat{\gauge}_2 \right )
\otimes_s \left ( \frac{dz_1}{z_1}  + \frac{dz_2}{z_2} \right ) \\
= & \G_{(z_3,\gauge_3)} ( \nablao )
\end{align*}
where in the last equality we used \eqref{gaugenabla1} with $(z,\gauge)
\rightarrow (z_3, \gauge_3) = (z_1 z_2, \gauge_1 + z_1^{-1} \gauge_2)$.
\end{proof}  

The connection $\nablao$ is defined by its action on $\gamma$ and
$\ellc$. Its action on  all other null metric hypersurface fields is then fully determined. The result is the following
\cite{Mars2020}[Lemma 4.8]
\begin{align}
\nablao_{a} \gamma_{bc} & = - \ell_b \U_{ac} - \ell_c \U_{ab}, \label{nablaogamma} \\
\nablao_a \ell_b & = \F_{ab} - \ll \U_{ab}, \label{nablaoll}\\
\nablao_a \vecn^b & =  \vecn^b \sone_a  + P^{bf} \U_{af},
\label{nablaon}  \\
\nablao_{a} P^{bc} & = - \left ( \vecn^b P^{cf} + \vecn^c P^{bf} \right ) \F_{af} - \vecn^b \vecn^c \nablao_a \ll. \label{nablaoP} 
\end{align}

So far we have considered null metric hypersurface data. For a  given ruled 
null manifold $(\N,\gamma)$ the correspondence in Proposition
\ref{correspondence} establishes that on $(\N,\gamma)$ we can define a collection of connections $\{\nablao\}$ related to each other by the gauge group $\G$ according to \eqref{gaugenablao}. The question is whether we can do a similar construction in the more general case of null manifold data.

Consider $(\N,\gamma)$ a null manifold. As already said, there exists a double covering $f: \dN \rightarrow \N$ such that $(\N,\gamma)$ induces by pull-back a ruled null manifold structure $(\dN,\dgamma)$. Hence, on $(\dN,\dgamma)$ we can construct an equivalence class of null metric hypersurface
$\dmetdata$ as well as the gauge group $\dG = \FF^{\star} (\dN) \times
\XX(\dN)$. Thus the double covering admits a collection of
connections
\begin{align*}
 \Big  \{ \dG_{(\dz,\dgauge)} ( \dnablao), \quad \dz \in \FF^{\star} (\dN),
  \dgauge \in \XX(\dN) \Big \}.
\end{align*}
The family  of null metric hypersurface data is far too large from the point of view of $\N$, in the sense that most of them will not descend to $\N$. Our aim now is to restrict this class in such a way that a suitable
geometry can be defined on $(\N,\gamma)$. Since $f$ is a local diffeomorphim,
at any point $\dq \in \dN$ we can define a map
$f_{\star} |_{\dq} : T^{\star}_{\dq} \dN \rightarrow T^{\star}_{f(\dq)} \N$ by
$f_{\star}|_{\dq} := f^{-1}{}^{\star} |_{f(\dq)}$
where $f^{-1} : f(\dU) \rightarrow \dU$ is the inverse of $f$
restricted to a suitable small neighbourhood $\dU$ of $\dq$.  For any two
points $\dq_1, \dq_2 \in \dN$ we write $\dq_1 \sim \dq_2$ iff $f(\dq_1)= f(\dq_2)$. We say
that null metric data $\dmetdata$ is {\bf compatible with the covering}
(or simply ``compatible'') iff the following two conditions hold:
\begin{itemize}
\item[(i)] For any pair of points $\dq_1 \sim \dq_2$  there exists $\epsilon \in \{ -1,1\}$ such that
  \begin{align}
    f_{\star}|_{\dq_2} (\dellc) = \epsilon f_{\star}|_{\dq_1} (\dellc),
    \label{defepsilon}
  \end{align}
\item[(ii)] There exists $\elltwo \in \FF(\N)$ such that $\delltwo := f^{\star} (\elltwo)$.
\end{itemize}
Note that for any function $\dhh \in \FF(\dN)$ the condition
$\dhh (\dq_1) = \dhh (\dq_2)$
for all pairs of point $\dq_1 \sim \dq_2$ is equivalent to the existence of a function $h \in \FF(\N)$ such that $\dhh = f^{\star} (h)$. Thus, 
condition (ii) is equivalent to $\delltwo(\dq_1) = \delltwo(\dq_2)$.
Note also that an immediate consequence of (i)
is that
\begin{align}
  f_{\star}|_{\dq_2} (\dupn) = \epsilon f_{\star} |_{\dq_1} (\dupn) \label{conddupn}
\end{align}
where we have called $\dupn$
the section in the radical normalized by $\dellc(\dupn)=1$ (cf. \eqref{prod2}).
The following proposition shows that this set of compatible data is not-empty, and identifies the gauge subgroup of $\dG$ which acts internally in this class.
The case when $(\N,\gamma)$ itself is a ruled manifold is trivial because then
$\dN$ is simply two disjoint copies of $\N$.  We therefore assume that $(\N,\gamma)$ is not a ruled manifold. In that case the manifold $\dN$ is connected, as we show next.
\begin{lemma}
  If the null manifold $(\N,\gamma)$ is not a ruled manifold, then its
  covering ruled manifold $(\dN,\dgamma)$ is connected.
\end{lemma}
\begin{proof}
  We use the standard fact the restriction of any covering to a connected component of the domain is still a covering. If $\dN$ had more than one connected component, then the restriction of $f$ to one such components, say $\dN_1$ would define a diffemorphism $f_1: \dN_1 \rightarrow \N$ (the restriction  of $f$ to {\em each} connected component of $\dN$ is a covering, so when  $\dN$ has more than one component the anti-image of $f_1$ must have precisely one element). Since $(\dN_1,\gamma|_{\dN_1})$ is a ruled null manifold, so it would be $(\N,\gamma)$, against hypothesis.
\end{proof}
\begin{proposition}
  \label{gaugecovering}
  Let $(\N,\gamma)$ be a null manifold which is not a ruled null manifold. Let  
$(\dN,\dgamma)$ be the ruled null manifold constructed from the double covering of $(\N,\gamma)$. Then the class of null metric hypersurface data compatible with the covering is not empty. Moreover
if ${\mathcal D}:= \dmetdata$ is compatible with the covering
then $\G_{(\dz, \dgauge)} (\D)$ is also compatible if and only if
  \begin{itemize}
  \item[(i)] There exists $z \in \FF^{\star} (\N)$ such that
    $\dz = f^{\star} (z)$.
  \item[(ii)] For any pair of points $\dq_1 \sim \dq_2$ it holds
    $\displaystyle{
    f_{\star} |_{\dq_2} (\dgauge) = \epsilon f_{\star} |_{\dq_1} (\dgauge)}$, where
    $\epsilon$ is the same as in \eqref{defepsilon}.
    \end{itemize}
    \end{proposition}
    \begin{proof}
    To show existence, choose a riemannian metric $g$ in $\N$ and define
    $\dg := f^{\star} (g)$. Select a nowhere zero section $\dX$ of the radical bundle $\widetilde{E} := \bigcup_{\dq \in \dN} \Rad_{\dgamma} |_{\dq}$ and define $\dupn
    = \frac{\dX}{\sqrt{\dg (\dX,\dX)}}$ as well as
    \begin{align*}
      \dellc := \dg ( \dupn, \cdot).
    \end{align*}
Choose also any smooth function
$\elltwo \in \FF(\N)$ and define $\delltwo := f^{\star} ( \elltwo)$.
It follows directly from Lemma \ref{conditionsmetdata}
that $\dmetdata$ defines null metric hypersurface data on $\dN$.  The construction is such that this data is compatible with the covering. Indeed, for any pair of related points $\dq_1 \sim \dq_2$ with image $q := f(\dq_i)$, the two vectors 
$f_{\star}|_{\dq_i} (\dupn)$ belong to the radical of $\gamma|_{q}$ and are unit with respect to $g$, so they are related by  a sign $\epsilon$. The claim follows
because $f_{\star} |_{\dq_i} ( \dellc) = g ( f_{\star}|_{\dq_i} (\dupn), \cdot)$.

To identify the subgroup of $\dG$ that maps compatible data into compatible data, we consider any two metric data
$\dD_1 = \dmetdataone,  \dD_2=\dmetdatatwo$
compatible  with the covering. We know by Corollary \ref{equivdata} that there exist $\dz \in \FF (\dN)$ and $\dgauge \in \XX(\dN)$ such that
$\G_{(\dz,\dgauge)} (\dD_1) = \dD_2$. Let $\dupn_1$ and $\dupn_2$ be the corresponding sections of the radical satisfying
$\dellc_1 (\dupn_1) = 1$ and
$\dellc_2 (\dupn_2) = 1$. They are related by $\dupn_1 = \dz \dupn_2$.
Select any two points $\dq_1 \sim \dq_2$ and let $\epsilon_1, \epsilon_2$ be the signs such that \eqref{conddupn} hold for $\dupn_1$ and $\dupn_2$ respectively. Then
\begin{align*}
  \dz |_{\dq_1} f_{\star} |_{\dq_1} (\dupn_2) =
  f_{\star}|_{\dq_1} (\dupn_1) = \epsilon_2 f_{\star} |_{\dq_2} (\dupn_1)
  =
  \epsilon_2 f_{\star} |_{\dq_2} (\dz \dupn_2)
  = \dz |_{\dq_2}   \epsilon_2 f_{\star} |_{\dq_2} (\dupn_2)
  = \dz |_{\dq_2} \epsilon_2 \epsilon_1
  f_{\star} |_{\dq_1} (\dupn_2)
\end{align*}
which implies $\dz|_{\dq_1} = \epsilon \dz |_{\dq_2}$ with $\epsilon :=
\epsilon_1 \epsilon_2 \in \{-1,+1\}$. However $\dz$ vanishes nowhere and $\dN$ is connected, so $\dz$ is either everywhere positive or everywhere negative. Thus $\epsilon =1$.
As we have already noted, $\dz|_{\dq_1} = \dz|_{\dq_2}$ for any pair $\dq_1 \sim \dq_2$ is equivalent
  to the existence of a function $z \in \FF^{\star} (\N)$ such that
  $\dz = f^{\star} ( z)$. This proves item (i) of the proposition.  Note that $\epsilon =1$ means also that $\epsilon_1
  = \epsilon_2$. We now use
  \begin{align}
    f_{\star}|_{\dq_2} (\dellc_i) = \epsilon_i f_{\star}|_{\dq_1} (\dellc_i),  \qquad
    i= 1,2. \label{ca}
      \end{align}
We define $q := f(\dq_1) = f(\dq_2)$.  Inserting the gauge transformation \eqref{gaugenablao},
  \begin{align*}
    f_{\star}|_{\dq_2} (\dellc_2) = \left \{ \begin{array}{l}
\displaystyle{     f_{\star} |_{\dq_2} \big ( \dz ( \dellc_1 +
      \dgamma(\dgauge, \cdot ))  \big ) =
       z |_{q} \big (
      f_{\star} |_{\dq_2} (\dellc_1)
        + \gamma( f_{\star}|_{\dq_2}  ( \dgauge) , \cdot ) \big )} \\ \\
\epsilon_2 f_{\star}|_{\dq_1} (\dellc_2)  =
    \epsilon_2 f_{\star}|_{\dq_1} \big ( \dz  ( \dellc_1 + \dgamma(\dgauge,\cdot) ) \big )  =
\epsilon_2    z |_{q} \big (
    f_{\star} |_{\dq_1} (\dellc_1) + \gamma( f_{\star}|_{\dq_1} (\dgauge), \cdot ) \big )
    \end{array} \right .
      \end{align*}
  Since $\epsilon_2 f_{\star} |_{\dq_2} (\dellc_1) = 
  f_{\star} |_{\dq_1} (\dellc_1)$ (here we use that $\epsilon_2 = \epsilon_1$), we conclude
  \begin{align*}
   \gamma(  f_{\star}|_{\dq_2} (\dgauge), \cdot ) 
  = \epsilon_1 \gamma( f_{\star}|_{\dq_1} (\dgauge), \cdot ) 
  \end{align*}
  which is equivalent to
  \begin{align}
   f_{\star}|_{\dq_2} (\dgauge) =
   \epsilon_1 f_{\star}|_{\dq_1} (\dgauge) + a X, \qquad a \in \mathbb{R}
   \label{condd}
  \end{align}
  where $0 \neq X \in \Rad_{\gamma} |_{q}$.
  Finally, we apply
  $\delltwo_i ( \dq_1) = \delltwo_i (\dq_2), i=1,2$
  \begin{align*}
    \delltwo_2 (\dq_2) = 
       \left \{ \begin{array}{l}
         z^2 |_{q} \big ( \delltwo_1 (\dq_2)
         + 2 \dellc_1  (\dgauge) |_{\dq_2}         + \dgamma  ( \dgauge,\dgauge) |_{\dq_2} 
 \big ) \\
 \\
          \delltwo_2 (\dq_1) =
         z^2 |_{q} \big ( \delltwo_1 (\dq_1)
 + 2 \dellc_1 (\dgauge) |_{\dq_1}        + \dgamma ( \dgauge,\dgauge) |_{\dq_1} 
 \big )
     \end{array} \right . 
      \end{align*}
  which means that
  \begin{align*}
    2 \dellc_1  (\dgauge) |_{\dq_2}         + \dgamma  ( \dgauge,\dgauge) |_{\dq_2}
=    2 \dellc_1  (\dgauge) |_{\dq_1}         + \dgamma  ( \dgauge,\dgauge) |_{\dq_1}. \end{align*}
  This equality can be equivalently written as
  \begin{align*}
2 f_{\star} |_{\dq_2} (\dellc_1) ( f_{\star} |_{\dq_2} (\dgauge))
+ \gamma ( f_{\star} |_{\dq_2} (\dgauge), f_{\star} |_{\dq_2} (\dgauge))  =
    2 f_{\star} |_{\dq_1} (\dellc_1) ( f_{\star} |_{\dq_1} (\dgauge))
    + \gamma ( f_{\star} |_{\dq_1} (\dgauge), f_{\star} |_{\dq_1} (\dgauge))
   \end{align*}
  which after inserting  \eqref{ca} with $i=1$ and  \eqref{condd} simplifies to
  \begin{align*}
    2 \epsilon_1 f_{\star} |_{\dq_1} (\dellc_1) ( a X) =0.
      \end{align*}
  The only solution to this equation is $a=0$, which concludes the proof.
  \end{proof}
  
  The main consequence of Proposition \ref{gaugecovering} is that a null manifold $(\N,\gamma)$ admits a class of well-defined torsion-free connections $\{\nablao\}$. The key property behind this fact is that
  the metric hypersurface connection remains invariant under
  the transformation defined by the gauge element $(z=-1,\gauge =0)$. 
  \begin{theorem}
    \label{ExistCon}
    Let $(\N,\gamma)$ be a null manifold. There exists a family
    of torsion-free connections
    $\mathcal Q  = \{ \nablao\}$ globally defined on $\N$.
        Moreover, on each non-empty open subset
    $U \subset \N$ where $(U,\gamma|_U)$ is a ruled null manifold, each element
    $\nablao  \subset \mathcal Q$ restricted to $U$ corresponds to the metric
    hypersurface  connection of a null metric hypersurface data
    of the form $\{U,\gamma|_U, \ellc, \elltwo\}$.
  \end{theorem}

  \begin{proof}
    If the null manifold $(\N,\gamma)$ is ruled, the class of connections
    $\mathcal Q$ is simply the class of metric hypersurface connections
    on $(\N,\gamma)$. Restricting a metric hypersurface data to a non-empty open
    subset $U$ still defines a null metric hypersurface data, and the restriction of the connection is obviously the $\nablao$ connection of the restricted null hypersurface data. So, we only need to worry about the case when $(\N,\gamma)$ is not a ruled null manifold. We know that the double cover $(\dN, \dgamma)$ is a ruled null manifold, so we can define the collection of
    null metric hypersurface connections $\{ \dnablao \}$. We fix a
    null metric hypersurface data $\dmetdata$ compatible with the covering
    and let $\dnablao$ be the corresponding metric hypersurface connection. We want to prove that
    $\dnablao$ descends to $\N$ and thus defines a connection which we call
    $\nablao$. To define $\nablao$ we select a sufficiently small neighbourhood $U$ of $q \subset \N$ so that
    $\pi^{-1} (U)$ has two connected components  $\dU_i$, $i=1,2$.  Let $\dq_i = \pi^{-1} (q) \cap \dU_i$, $i=1,2$ and $f_i  : \dU_i \rightarrow U$ be the restriction of $f$ to $\dU_i$. By construction $f_i$ are diffeomorphisms. We can use each $f_i$ to define a connection $\nablao_i$ on $U$ by transferring (with this diffeomorphism)
     $\dnablao$ in $\dU_i$ to $U$. Since $\dnablao$ is defined geometrically
    in terms of the metric hypersurface data, it follows that
    $\nablao_i$ is the null metric hypersurface connection of the 
    null metric hypersurface data
       \begin{align*}
      \{U, \gamma|_U, \ellc_i := f_i^{-1}{}^{\star} ( \dellc |_{\dU_i}),
        \elltwo_i :=\delltwo \circ f_i^{-1}\}.
    \end{align*}
        Each of them defines a null metric hypersurface
        data on the same  manifold $U$ and they share the tensor  $\gamma$, so
        by Corollary \ref{equivdata}
        they are related by a gauge transformation. By definition of null metric hypersurface data compatible with the covering, there exists $\epsilon \in \{-1,1\}$ such that (in principle the sign $\epsilon$ could depend on the point, but this cannot happen in the present context because $\ellc_1$ and $\ellc_2$ are smooth and nowhere zero covectors on $U$)
        \begin{align*}
      \ellc_1 = \epsilon \ellc_2, \qquad
      \elltwo_1 = \elltwo_2.
        \end{align*}
        The gauge element that transforms one to the other is $(z= \epsilon,
        \gauge =0)$. The transformation law for the metric hypersurface
        connection \eqref{gaugenablao} implies $\nablao_1 = \nablao_2=: \nablao$.
        So, $\dnablao$  defines the same connection $\nablao$ on $U$. It is clear that $\nablao$ defines a torsion-free connection on all of $\N$.

        The second claim in the theorem is immediate because if $U$ is such that
        $(U,\gamma)$ is a ruled manifold then
        $\pi^{-1} (U)$ has two connected components, and by construction
        $\nablao$ is the metric hypersurface data of the null metric hypersurface data $(U, \gamma|_U,  f_1^{-1}{}^{\star} ( \dellc |_{\dU_1}),
        \delltwo \circ f_1^{-1})$.
          \end{proof}

      \section{The energy-momentum map: algebraic properties}
\label{EMmap}
In this section we introduce a linear map on metric hypersurface data and study its properties. This map finds its motivation in the theory of matching of spacetimes across null boundaries. However, it has many more applications. Thus, for the moment we prefer to simply put forward its definition and explore the consequences. The connection with the matching problem  will be the subject of
Section \ref{shells}. The justification for the name of the map will also be described there.

So, we define a linear map $\tau$
called {\bf energy-momentum map}
that sends symmetric $(0,2)$-tensors to symmetric $(2,0)$-tensors on any null metric hypersurface data set.

\begin{definition}
        Let $\metdata$ be null metric hypersurface data. At any point $p \in \N$ we define the {\bf energy-momentum map} at $p$ as the map
\begin{align}
\tau|_p : \Sd(T_p \N)  & \longrightarrow  \Su(T_p \N)   \nonumber \\
 V & \mapsto  \tau|_p (V)^{ab} \defi 
\left ( \left ( n^a P^{bc} + n^b P^{ac} \right) n^d V_{cd} -
 P^{ab} n^c n^d  V_{cd} - n^a n^b  P^{cd} V_{cd} \right )|_p. \label{expressiontau}
\end{align}
\end{definition}
In  this section we shall study the algebraic properties of the energy-momentum map and in the following section its PDE properties. Throughout this section everything is evaluated at a single point $p$. For simplicity we drop any explicit  reference to $p$, e.g. we write $\tau$ instead of $\tau|_p$, except when this may cause confusion. We shall also write $\Sd$ instead of $\Sd(T_p \N)$ and
 $\Su$ instead of $\Su(T_p \N)$.

The map $\tau$ is linear and can therefore be written in terms of
a four-contravariant tensor $S^{abcd}$ as $\tau(V)^{ab} =  S^{abcd} V_{cd}$, where
\begin{align}
 S^{abcd} = & \frac{1}{2} \left ( n^a P^{bc} n^d 
+ n^b P^{ac} n^d + n^a P^{bd} n^c + n^b P^{ad} n^c 
- 2 n^a n^b P^{cd} - 2 n^c n^d P^{ab}  \right ). \label{defS}
\end{align}
This tensor has the following explicit algebraic properties
\begin{eqnarray*}
S^{abcd} = S^{bacd}, \quad \quad
S^{abcd} = S^{abdc}, \quad \quad
S^{abcd} = S^{cdab},
\end{eqnarray*}
so that $S$ defines a symmetric linear map
$S : \Sd \otimes \Su \longrightarrow \mathbb{R}$
(i.e. it is a symmetric two-covariant tensor
on $\Sd$) by $S(V_1,V_2) = S^{abcd} (V_1)_{ab} (V_2)_{cd}$.
Using the natural identification between $\Su$ and
the dual space $(\Sd)^{\star}$,  we can write the map $\tau$ as
\begin{eqnarray*}
\tau(V) = S (\cdot, V)
\end{eqnarray*}
and view $\tau$ as a map $\tau: \Sd \longrightarrow (\Sd)^{\star}$.
As a consequence of the symmetry of $S$, the dual map $\tau^{\star}$
(also called transpose) is the same as $\tau$.
Recall the standard relationship (see e.g. \cite{Warner})
between the kernel
of a linear map $f: {\mathcal F}_1 \longrightarrow {\mathcal F}_2$
and the range of its dual $f^{\star}: {\mathcal F}^{\star}_2 
\longrightarrow {\mathcal F}_1^{\star}$
\be
\mbox{Ker} (f) = \mbox{Ran} (f^{\star})^{\bot}
\en
where obviously the orthogonal is in the sense of dual spaces. 
If we denote by $\K_0 \subset \Sd$ the kernel of $\tau$, it
follows
\begin{align}
\mbox{Ran} (\tau)  = \mbox{Ran}(\tau^{\star} ) =
\mbox{Ker} (\tau)^{\bot} = \K_0^{\bot} = & \{ \Tau \in \Su 
\mbox{ such that } 
\Tau^{ab} V_{ab}=0 \label{RanKer} \\
& \mbox{ for all } V \in \Sd \mbox{ satisfying } 
\tau(V)=0\}. \nonumber
\end{align}
Our aim is to identify the kernel and the
range of $\tau$ as well as the pre-image of any element 
$\Tau \in \mbox{Ran} (\tau)$.

Using the fact that $\upn$ belongs to the radical of $\gamma$ together with
\eqref{prod4} it is immediate to show that
\begin{align}
  & \tau_{\gamma}(V)^{a}_{\,\,\,c} \defi
 \tau(V)^{ab} \gamma_{bc} =  
 \vecn^a \vecn^b  V_{cb} - \delta^a_{c} \vecn^b \vecn^d  V_{bd}.
 \label{T(V)}
\end{align}
Taking the trace of this expression yields
\begin{align}
& \tau_{\gamma}(V)^{a}_{\,\,\,a} = -(\n-1) \vecn^c \vecn^d V_{cd}. 
  \label{alg1}
\end{align}
We also compute $\tau(V)^{ab} \ell_b$ which, after using \eqref{prod3},  becomes
\begin{align}
  \tau(V)^{ab} \ell_a = P^{bc} V_{cd} \vecn^d - n^b P^{cd} V_{cd} \label{alg3}
\end{align}
and its contraction with $\gamma_{bc}$ gives, from \eqref{prod4},
\begin{align}
& \tau(V)^{ab} \gamma_{bf} \ell_{a} =  V_{fd} \vecn^d  - \ell_f \vecn^c \vecn^d V_{cd}. \label{alg2}
\end{align}
We can now characterize the range and the kernel of  $\tau$
and identify the pre-image of any element in the range.
\begin{proposition}
\label{Rangetau}
Let $\metdata$ be null metric hypersurface data
of dimension $\n \geq 2$. Let $p \in \N$
and $\tau$ the energy-momentum map at $p$. Then 
the kernel $\K_0$ of 
$\tau$ is  
\begin{eqnarray}
\K_0 = \{ V \in S^0_2(T_p \N) ; \quad  V_{ab} \vecn^b  =0, \,
V_{ab} P^{ab} = 0  \}. \label{kernel}
\end{eqnarray}
Moreover, a symmetric tensor $\Tau^{ab}$ at $p$ belongs to $\mbox{Ran} (\tau)$
if and only if there exists a vector $X \in T_p \N$ and a scalar
$\Q \in \mathbb{R}$ such that 
\begin{eqnarray}
  \Tau = 2 \upn \otimes_s  X - \Q P.
\label{T(X,Q)}
\end{eqnarray}
Given $\Tau \in \mbox{Ran}(\tau)$, the vector $X$ and scalar $\Q$ in this decomposition are uniquely defined, and 
the anti-image of
$\Tau$ is given by
\begin{align*}
  \tau^{-1} (\Tau) = \{
  2 \ellc \otimes_s  \uwidehat{X} + \Q \ellc \otimes \ellc
  - \frac{1}{\n-1} \left ( \Q \ll + 2 \ellc(X)  \right ) \gamma
+ \VH, \quad \VH \in \K_0 \}
\end{align*}
where $\uwidehat{X}\defi \gamma(X, \cdot)$.
\end{proposition}

\begin{proof}
  Let $V \in \mbox{Ker}(\tau)$, namely $\tau(V)=0$. Then, using the condition $
  \n \geq 2$, (\ref{alg1}) implies 
$\vecn^c \vecn^d  V_{cd} =0$ which inserted in \eqref{alg2}  gives
$\vecn^b V_{cb}  = 0$ and then \eqref{alg3} implies $P^{cd} V_{cd}=0$
(recall that $\upn$ cannot vanish anywhere). Conversely,
any symmetric tensor $V_{ab}$ satisfying $V_{ab} \vecn^a =0$ and
$P^{ab} V_{ab}=0$ lies in the kernel of $\tau$, as $\tau(V)=0$ follows directly
from (\ref{expressiontau}). So the kernel 
is the vector subspace $\K_0$ defined in (\ref{kernel}). 
From \eqref{RanKer} we also get 
\begin{align*}
\mbox{Ran} (\tau ) = \{ \Tau \in S^2_0 \mbox{ such that } \Tau^{ab} V_{ab}=0
\mbox{ for all } V_{ab} \mbox{ satisfying } V_{ab} \vecn^b =0, 
P^{ab} V_{ab} =0 \}.
\end{align*}
The dimension of $S^0_2$ is obviously $\n(\n+1)/2$ and hence
$\mbox{dim} (\K_0) = \n (\n-1)/2 -1$ because the subspace $\K_0$ is defined by $\n+1$ linearly independent conditions. Thus, $\mbox{dim} ( \mbox{Ran}(\tau)) = \n +1$. Note first that, by the very definition of $\K_0$, given any
$X \in T_p \N$ and $\Q \in \mathbb{R}$ one has
\begin{eqnarray*}
\left ( \vecn^a X^b + \vecn^b X^a - \Q P^{ab} \right ) V_{ab}=0
\end{eqnarray*}
for any $V \in \K_0$.  Let us now check that a tensor of the type $\Tau^{ab} = \vecn^a X^b + X^b \vecn^a - \Q P^{ab}$ is identically zero  if and only if $X^a =0, \Q=0$. Sufficiency is obvious. To show necessity we contract $\Tau^{ab}$ with $\gamma_{ac} \gamma_{bd}$ to get $\Q \gamma_{cd} =0$, so $\Q=0$ because
$\gamma_{cd} \neq 0$ (here we use again that $\n \geq 2$), and now $X^a=0$ follows at once because  a tensor product of two vector
$X_1 \otimes_s X_2$ can vanish only if one of the vectors vanishes. Consequently,
the set of tensors of the form $\vecn^a X^b + \vecn^b X^a - \Q P^{ab}$
is a vector
subspace of dimension $\n+1$, so it must agree with $\mbox{Ran}(\tau)$. Note that the argument above also proves that a given $\Tau \in \mbox{Ran}(\tau)$ defines unique  $X^a$, $\Q$ according to the expression
$\Tau^{ab} = \vecn^a X^b + \vecn^b X^a - \Q P^{ab}$, so the uniqueness claim in the proposition holds. 

It only
remains to find the general solution of $\tau(V) = \Tau$ for
\begin{eqnarray*}
\Tau^{ab} = \vecn^a X^b + \vecn^b X^a - \Q P^{ab}.
\end{eqnarray*}
It suffices to find a particular solution. Recalling that
$\uwidehat{X} \defi \gamma(X, \cdot)$ we look for solution of the
form
\begin{eqnarray*}
V_{ab} = \ell_a \uwidehat{X}_b + \ell_b \uwidehat{X}_a + \aa \ell_a \ell_b + \bb \gamma_{ab}
\end{eqnarray*}
where $\aa$ and $\bb$ are quantities to be determined.
Since $\gamma_{ab} \vecn^b=0$, $\uwidehat{X}_a \vecn^a = 0$,
$\vecn^a \ell_a = 1$ and $P^{ab} \gamma_{ab} = 
\n-1$ we have $\Tau^{a}_{\,\,\,a} \defi \Tau^{ab} \gamma_{ab} = - (\n-1) \Q$.
Given that $V_{cd} \vecn^c \vecn^d = \aa$,  equation (\ref{alg1}) fixes $\aa = \Q$.
We impose now \eqref{alg3}. A direct calculation which uses
$P^{ab}  \uwidehat{X}_b = P^{ab} \gamma_{bc} X^c = X^a - \vecn^a (\ell_c X^c)$ yields
simply
\begin{align*}
\vecn^a \left ( 2 \ell_c X^c + \Q \ll + \bb (\n -1 ) \right )= 0
\quad \Longrightarrow \quad \bb = - \frac{1}{\n-1} \left ( \Q \elltwo + 2 X^c \ell_c
\right )
  \end{align*}
and $\bb$ has been determined. It is now a matter of substitution into
(\ref{expressiontau})
to check that
\begin{eqnarray}
V^{(p)}_{\Tau} \defi 2 \ellc \otimes_s \uwidehat{X} + \Q \ellc \otimes  \ellc
- \frac{1}{\n-1} \left ( \Q \ll + 2 \ellc(X) \right ) \gamma
\label{particular}
\end{eqnarray}
satisfies $\tau(V^{(p)}_{\Tau}) = 2 \upn \otimes_s X -  \Q P$,  as claimed. 
\end{proof}

Our next result is a simple consequence of Proposition \ref{Rangetau}
and will
be used repeatedly. Before stating the proposition it is convenient to introduce the following terminology:
\begin{terminology}
\label{TerVs}
  For any symmetric $(0,2)$-tensor $V$ we define the following scalars, covector and vector
  \begin{align}
    \Z := P^{ab} V_{ab}, \qquad \Vnn := V_{ab} \vecn^a \vecn^b, \qquad \Vnp_{b} := V_{ab} \vecn^a, \qquad 
    \X^a := P^{ab} \Vnp_b - \frac{1}{2} \Z \vecn^a. \label{defTer}
   \end{align}
  Moreover, for any vector $X$ we use an underhat to denote the covector
       $\uwidehat{X} := \gamma(X, \cdot)$ and let $\unl{X}$ be defined by the 
decomposition
  \begin{align*}
    X = \ellc(X) \upn + \unl{X}.
  \end{align*}
  By construction the vector $\unl{X}$ satisfies  $\ellc(\unl{X})=0$. 
   \end{terminology}

The following expression follows directly from the definitions
\begin{align}
  \uwidehat{\X}_a = \gamma_{ab} \X^b =
  \gamma_{ab} \big ( P^{bc} \Vnp_c - \frac{1}{2} \Z \vecn^b \big) =
  \Vnp_a - \Vnn \ell_a
  \label{unX}
\end{align}
where in the last equality we used \eqref{prod1} and \eqref{prod4}. Moreover,
  \begin{align}
\ellc(\X) = P^{ab} \ell_a \Vnp_b  - \frac{1}{2} \Z
 = - \ll \vecn^b \Vnp_b - \frac{1}{2} \Z =
- \Vnn \ll  - \frac{1}{2} \Z, \label{ellcX}
\end{align}
so the vector $\unl{\X}$ takes the explicit form
\begin{align}
  \label{Xl}
\unl{\X}^{\,a} = P^{ab} \Vnp_b + \Vnn \ll \vecn^a 
\end{align}
and we can also write
\begin{align}
  \X = - \big ( \Vnn \ll + \frac{1}{2} \Z \big ) \upn + \unl{\X}.
  \label{tangent-normal}
\end{align}
We note for later use that  \eqref{unX} and \eqref{Xl} imply at once
\begin{align}
  \unl{\X}{}^a = P^{ab} \uwidehat{\X}_b.
  \label{raise}
\end{align}
\begin{proposition}
\label{Decom}
Let $\metdata$ be null metric hypersurface data, $p \in \N$ and $V \in \Sd(T_p \N)$. Then $V$ can be decomposed uniquely as
\begin{align}
  V = 2 \ellc \otimes_s  \uwidehat{\X} + \Vnn \left ( \ellc \otimes \ellc
+ \frac{1}{\n-1} \ll \gamma \right ) 
  + \frac{\Z}{\n-1}  \gamma + \VH,
\label{Form3}
\end{align}
where $\VH$ is a tensor in the kernel of $\tau$, i.e. $\VH \in \K_0$.
\end{proposition}
\begin{remark}
Using \eqref{unX} this decomposition can also be written as
\begin{equation}
V = 2 \ellc \otimes_s  \W - \Vnn \ellc \otimes  \ellc
+ \frac{1}{\n-1} \left ( \Vnn \ll + \Z \right ) \gamma + \VH.
\label{Form2}
\end{equation}
\end{remark}

\begin{remark}
\label{orthogonalcase}
If $V$ satisfies $\Vn =0$, then both $\uwidehat{\X}$ and $\Vnn$
vanish identically and the decomposition simplifies to
\be
V = \frac{1}{\n-1} (\Z) \gamma + \VH, \qquad \VH \in \K_0.
\en
\end{remark}
\begin{corollary}
  \label{decomprod}
  Let $\bm{\omega_i}$, $i=1,2$ be a pair of two-forms satisfying $\bm{\omega_i} (\upn) = 0$. Then $V := \bm{\omega_1} \otimes_s \bm{\omega_2}$ decomposes
  as
\be
V = \frac{P(\bm{\omega_1},\bm{\omega_2})}{\n-1}  \gamma +
\left ( \bm{\omega_1} \otimes_s \bm{\omega_2}\right )^H. 
\en
\end{corollary}
\begin{proof}[Proof of Proposition \ref{Decom}]
  With the definitions of $\X$ and $\Vnn$ as given in Terminology \ref{TerVs}, it is immediate that  $\tau(V)$ (cf. \eqref{expressiontau}) takes the form
  \be
\tau(V) = 2 \upn \otimes_s \X  - \Vnn P.
\en
So,  directly
 from Proposition \ref{Rangetau} with $X^a = \X^a$ and $\Q =  \Vnn$
we can write
\begin{equation*}
V = 2 \ellc \otimes_s \uwidehat{\X} + \Vnn  \ellc \otimes \ellc 
- \frac{1}{\n-1} \big ( \Vnn \ll + 2 \ellc (\X)  \big ) \gamma
+ \VH,
\end{equation*}
and the decomposition is unique. By \eqref{ellcX} this decomposition  
takes the form \eqref{Form3}.
\end{proof}
By construction, the decomposition in Proposition \ref{Decom} is well adapted to the
energy-momentum map. For later use we quote how $\tau$ acts on a
tensor $V$ decomposed according to the lemma.
\begin{corollary}
  \label{DecomImage}
  Let $\metdata$ be null metric hypersurface data, $p \in \N$ and $V \in \Sd(T_p \N)$. Decompose $V$ according to Proposition \ref{Decom}. Then, its image under the energy-momentum map is
  \begin{align*}
    \tau(V)  = 2 \upn \otimes_s \X   - \Vnn P
  \end{align*}
  where $\Vnn$ and $\X$ are defined in Terminology \ref{TerVs}. In terms
  of the vector $\unl{\X}$ it reads
  \begin{align}
    \tau(V) = - (\Z)  \upn \otimes \upn + 2 \upn \otimes_s \unl{\X}
- \Vnn \big ( P + 2 \ll \upn \otimes \upn \big ). \label{decommu}
  \end{align}
  \end{corollary}
An interesting consequence of Proposition \ref{Decom} is that the
map that sends $V$ into $V-\VH$  is a projector. Before proving this fact we note that decomposition \eqref{Form3} has introduced naturally a symmetric tensor, namely
$\tgamma:= \ellc \otimes \ellc + \frac{\ll}{\n-1} \gamma$. Let us establish its basic properties
\begin{lemma}
  \label{tgamma}
  The symmetric $(0,2)$-tensor $\tgamma:= \ellc \otimes \ellc + \frac{1}{\n-1} \ll\gamma$ satisfies
  \begin{align*}
        \iota_n \tgamma = \ellc, \qquad
    \mbox{tr}_P \tgamma = 0.
  \end{align*}
    \end{lemma}
\begin{proof}
  The first is immediate from $\gamma(n, \cdot)=0$ and $\ellc(n)=1$. The second holds because, from $\mbox{tr}_P \gamma = \n-1$,
  \begin{align*}
    P^{ab} \Big ( \ell_a \ell_b + \frac{1}{\n-1} \ll \gamma_{ab} \Big )
    \stackrel{\eqref{prod3}}{=} - \ll \vecn^b \ell_b +  \ll = 0.
\end{align*}
\end{proof}
We can now introduce the projector $\P$ that will play an important role in the rest of the paper.
\begin{proposition}
\label{defProjector}
The map $\P: \Sd(T_p \N) \longrightarrow \Sd(T_p \N)$ defined by
\begin{align}
  \P(V) & = 2 \ellc \otimes_s \uwidehat{\X}  +  \Vnn \tgamma 
+ \frac{1}{\n-1} \left (\Z \right ) \gamma \label{Proj} 
\end{align}
is a projector, i.e. a linear map satisfying $\P \circ  \P = \P$.
\end{proposition}

\begin{proof}
  It suffices to prove that
  $\mbox{tr}_P (\P(V)) = \Z$
  and $ \iota_n \P(V) = \Vn$
  because then
  also $\Q_{\P(V)} = \Vnn$ and
  $\uwidehat{(\overline{P(V)})} = \uwidehat{\X}$, so inserting into 
  \eqref{Proj} with $V \longrightarrow \P(V)$ yields $\P (\P(V)) = \P(V)$ at once.

 To establish $\mbox{tr}_P (\P(V)) = \Z$, and given that
 $\P(V)$ is the sum of three terms,
 it suffices to check  that
   $\mbox{tr}_P \gamma = \n -1$, $\mbox{tr}_P \, \tgamma =0$ and
  $\mbox{tr}_P (\ellc \otimes_s \uwidehat{\X} ) =0$. The first one is clear, the second has been established in Lemma \ref{tgamma} and the third holds true because of
  \begin{align*}
    P^{ab} \ell_{(a} \uwidehat{\X}_{b)}
    =    P^{ab} \ell_a \uwidehat{\X}_b = - \ll \vecn^b \uwidehat{\X}_b =0.
  \end{align*}
  Now, the property $\iota_n \P(V) = \Vn $
  follows from expression \eqref{unX} if $\iota_n \gamma = 0$, $\iota_n \tgamma =\ellc$ and
  $\iota_n (2 \ellc \otimes_s \uwidehat{\X}) = \uwidehat{\X}$. The first one is obvious, the second one has been proved in Lemma \ref{tgamma} and the last one is
  \begin{align*}
    2 \vecn^a \ell_{(a} \uwidehat{\X}_{b)} = \vecn^a \ell_a \uwidehat{\X}_b = \uwidehat{\X}_b.
  \end{align*}
  \end{proof}
\begin{remark}
\label{propertiesProjector}
Inserting the explicit form of $\uwidehat{\X}$, $\Vnn$ and $\Z$ in $\P(V)$
it follows that $\P$ can be written in the following form
\begin{align}
\P(V)_{ab} = \P_{ab}^{\,\,\,\,\,\,cd} V_{cd}, \quad
\P_{ab}^{\,\,\,\,\,\,cd} := 
2 \ell_{(a} \delta^{(c}_{b)} \vecn^{d)} - \vecn^c \vecn^d \ell_a \ell_b
+ \frac{1}{\n-1} \left ( \ll \vecn^c \vecn^d + P^{cd} \right ) \gamma_{ab}.
\label{defTensorP}
\end{align}
The following properties of $\P$ are either immediate or easily checked
\begin{itemize}
\item[(i)] $\P(V)=0$ for  all $V \in \K_0$.
\item[(ii)] $\P(\gamma) = \gamma$.
\item[(iii)] $\P( \ellc \otimes_s \bm{\omega} ) = \ellc
  \otimes_s \bm{\omega}$, for any covector
  $\bm{\omega} \in T^{\star}_p \N$.
\item[(iv)] Let $\bm{\omega_1}, \bm{\omega_2} \in T^{\star}_p \N$ be orthogonal to $\upn$, i.e. $\bm{\omega_1}(\upn) =
\bm{\omega_2}(\upn) = 0$. Then
\be
\P(\bm{\omega_1} \otimes_s \bm{\omega_2} ) =   \frac{1}{\n-1} P(\bm{\omega_1},
\bm{\omega_2} ) \gamma.
\en
\end{itemize}

\end{remark}

The decomposition in Proposition \ref{Decom} 
can be written as $V = \P(V) + \VH$. 
The following corollary is an immediate consequence of the uniqueness of the
decomposition together with Proposition \ref{defProjector} (we use $\I$ for the identity map).
\begin{corollary}
\label{I-P}
The map $\I - \P : \Sd(T_p \N)
\longrightarrow \Sd(T_p \N)$ is a projector. Moreover $\mbox{Ran}(\I-\P) = \K_0$ and
the restriction of $\I-\P$ to $\K_0$ is the identity map.
\end{corollary}

\subsection{Gauge properties of the energy-momentum map}

In order to describe the gauge behaviour of the map $\tau$ it is convenient
to introduce the notion of gauge weight.

\begin{definition}
  A tensor $T$ at  point $p$  is said to be of {\bf gauge weight} $q \in
  \mathbb{Z}$ if under a gauge transformation with parameters $(z,\gauge)$
  it transforms
  as
  \begin{align}
    \G_{(z,\gauge)} (T) = z^q T. \label{gaugeq}
  \end{align}
\end{definition}
This notion is well-defined because if \eqref{gaugeq} holds
then
\begin{align*}
  \G_{(z_1, \gauge_1)} \circ
  \G_{(z_2, \gauge_2)} (T) =
  \G_{(z_1, \gauge_1)} ( z_2^q T ) =
  z_1^q z_2^q T = (z_1 z_2)^q T = \G_{(z_1,\gauge_1) \cdot (z_2, \gauge_2)} (T).
\end{align*}
We will refer to the gauge weight simply as weight when there is no possible misunderstanding. Two examples of tensors with a well-defined gauge weight are the
tensor $\gamma$ (which has weight zero, i.e. it is gauge invariant) or the
vector $\upn$ which has gauge weight $q=-1$.

The next lemma shows that the map $\tau$ has a good behaviour in terms of weight (this result extends a discussion in  \cite{MarsGRG} in the context of shells, see also Sect. \ref{shells} below).
\begin{lemma}
  \label{weighttau}
  Let $V$ be a symmetric $(0,2)$-tensor of gauge weight $q$.
  Then
  $\Tau := \tau(V)$ has gauge weight $q-2$.
\end{lemma}

\begin{proof}
  Since $\Tau^{ab} = S^{abcd} V_{cd}$ the statement of the lemma is equivalent to showing that $\G_{(z,\gauge)} (S) = z^{-2} S$, i.e. that $S$ has gauge weight $-2$. From \eqref{defS} we may write
  \begin{align}
    S^{abcd}  = 2 \vecn^{(a} P^{b)(c} \vecn^{d)} - \vecn^a \vecn^b P^{cd} - \vecn^c \vecn^d P^{ab}. \label{ExpS}
  \end{align}
  From the transformation law of $P$ and $n$ given in \eqref{Pprime} it follows
  \begin{align*}
    \G_{(z,\gauge)} ( \vecn^a P^{bc} \vecn^d ) =
    \frac{1}{z^2} \left ( \vecn^a P^{bc} \vecn^d - \vecn^a \gauge^b \vecn^c \vecn^d - \vecn^a \gauge^c \vecn^b \vecn^d \right ),
  \end{align*}
  so
  \begin{align*}
        \G_{(z,\gauge)} ( \vecn^{(a} P^{b)(c} \vecn^{d)} ) =
      \frac{1}{z^2} \left ( \vecn^{(a} P^{b)(c} \vecn^{d)} - \vecn^{(a} \gauge^{b)} \vecn^c \vecn^d - \vecn^a \vecn^b \gauge^{(c}  \vecn^{d)} \right )
  \end{align*}
  and the property $\G_{(z,\gauge)} (S) = z^{-2} S$ follows from \eqref{ExpS} and
  \eqref{Pprime}.
\end{proof}
The property of $V$ having gauge weight $q$ can be stated equivalently in terms
of the decomposition in Proposition \ref{Decom}. The result is as follows
\begin{lemma}
  A tensor $V \in \Sd(T_p \N)$ has gauge weight $q$ if and only if the terms
  $\Vnn$, $\Z$, $\X$, $\VH$ in the decomposition \eqref{Form3} satisfy
  \begin{align}
    \G_{(z,\gauge)} (\Vnn) &= z^{q-2} \Vnn, \quad
      \G_{(z,\gauge)} (\Z) = z^q \left ( \Z - 2 V(\upn,\rig) \right ), \quad 
    \G_{(z,\gauge)} (\X) = z^{q-1} \left ( \X - \Vnn \rig \right ), \label{gaugecomps1} \\
    \G_{(z,\gauge)} (\VH) & = z^{q} \Big ( \VH +
                            \uwidehat{\rig} \otimes_s\left ( \Vnn \uwidehat{\rig}
                            + 2 \Vnn \ellc - 2 \iota_{\upn} V \right )
                            - \frac{1}{\n-1} \left (
                            \Vnn \left ( 2 \ellc(\rig) + \gamma(\rig,\rig)
                            \right ) - 2 V(\upn, \rig) \right ) \gamma \Big ).
\label{gaugecomps2}
                              \end{align}
\end{lemma}
\begin{proof}
Assume first that $V$ has gauge weight $q$. Then
$\G_{(z,\gauge)} (\iota_\upn V) = z^{q-1} \iota_{\upn} V$ 
as a consequence 
of the definition of $\iota_\upn V$ and \eqref{Pprime}. The transformation 
$\G_{(z,\gauge)} (\Vnn) = z^{q-2} \Vnn$ follows at once.  
The transformation of $\Z$ is also immediate from its definition and
\eqref{Pprime}. Now, using a prime to denote gauge transformed objects,
\begin{align*}
\X{}^{\prime}{}^a = P'{}^{ab}  (\iota_n V)'_b- \frac{1}{2} (\Z)' \vecn^{\prime}{}^a 
= (P^{ab} - \vecn^a \rig^b - \vecn^b \rig^a ) z^{q-1} (\iota_{\upn} V)_b
- \frac{1}{2} z^{q-1} \left ( \Z - 2 V(\upn,\rig) \right ) \vecn^a,<
\end{align*}
which is the third expression in \eqref{gaugecomps1} after cancelling terms. To prove 
\eqref{gaugecomps2}, as well as the converse in the lemma,
we insert the transformations \eqref{gaugecomps1} in the decomposition \eqref{Form2}
of $V'$. This yields the identity
\begin{align}
V' =  z^q \Big ( &2 ( \ellc + \uwidehat{\rig} ) \otimes_s \iota_{\upn} V
- \Vnn (\ellc + \uwidehat{\rig} )\otimes_s (\ellc + \uwidehat{\rig} )
\nonumber  \\ & + \frac{1}{\n-1} \left ( \Vnn ( \ll + 2 \ellc(\rig) + \gamma(\rig,\rig) )
+ \Z - 2 V(\upn,\rig) \right ) \gamma \Big ) + (\VH)^{\prime} \label{inter}
\end{align}
Now, if we assume that $V$ has gauge weight $q$ then we can replace $V' $ by $z^q V$ and insert the decomposition \eqref{Form2} in the left-hand side. This yields \eqref{gaugecomps2} after cancelling terms. Finally, if we assume
that \eqref{gaugecomps2} holds and replace this expression in the right-hand side of \eqref{inter} we conclude by a direct computation that $V' = z^q V$, i.e.
that $V$ has gauge weight $q$. 
\end{proof}

\begin{remark}
  It is easy to check that  $\G_{(z,\rig)} \VH$ as given in \eqref{gaugecomps2}
  also belongs to the kernel $\K_0$.
\end{remark}

It is natural to ask about the converse to Lemma \ref{weighttau}. Specifically, assume that the tensor
$\Tau = \tau(V)$ has gauge weight $r$, what is  then the possible gauge behaviour of $V$? And, in particular, is there any $V_{\Tau}$ with gauge weight $r+2$
satisfying $\tau(V_{\Tau}) = \Tau$?

We already know that any tensor $\Tau$ in the range of $\tau$ can be
parametrized uniquely by a vector $X$ and a scalar $\Q$ by
\begin{align}
  \Tau = 2 \upn \otimes_s X - \Q P. \label{decomTau}
\end{align}
We compute
\begin{align*}
  \G_{(z,\gauge)} (\Tau) & = 2 \G_{z,\gauge} (n) \otimes_s \G_{(z,\gauge)} (X) -
  \G_{(z,\gauge)} (\Q) \G_{(z,\gauge)} (P) \\
  & = 2 z^{-1} \upn \otimes_s \G_{(z,\gauge)} (X)
  - \G_{(z,\gauge)} (\Q) \left ( P  - 2 \upn \otimes_s \gauge \right ) \\
  & = z^r \left ( 2 \upn \otimes_s X - \Q P \right )
\end{align*}
the last equality being true if and only if $\Tau$ has  gauge weight $r$. We have already proved that $2 \upn \otimes X_1 - \Q_1 P$ vanishes if and only if $X_1 = 0$ and $\Q_1=0$. Thus, the equality above dictates the gauge behaviour of $X$ and $\Q$. Specifically we have
\begin{lemma}
  \label{GaugeXQ}
  A tensor $\Tau \in \mbox{Ran} (\tau)$ has weight $r$ if and only if the components $X$, $\Q$ in the decomposition  \eqref{decomTau} have gauge behaviour
  \begin{align*}
    \G_{(z,\gauge)} (X) = z^{r+1} \left ( X - \Q \gauge \right ), \qquad \quad
    \G_{(z,\gauge)} (\Q) = z^{r} \Q.
  \end{align*}
  \end{lemma}

  In particular, we see that when all elements in $\mbox{Ran}(\tau)$ have well-defined weight $r$ the vector subspace in the range of $\tau$
defined by the 
condition $Q=0$ is gauge invariant and hence well-defined. To address the question of whether we can select an element $V_{\Tau} \in \tau^{-1} (\Tau)$ with gauge weight $r+2$ we recall that we have already found an explicit particular solution
of $\tau(V) = \Tau$, namely the tensor $V^{(p)}_{\Tau}$ defined in \eqref{particular}. Let us therefore ask first what is the gauge behaviour of $V^{(p)}_{\Tau}$.
This tensor is explicitly defined in terms of $X$ and $\Q$, so finding its gauge transformation is simply a matter of direct computation. Using as before a prime to denote any gauge transformed object, we get
\begin{align}
  \G_{(z,\gauge)} (V^{(p)}_{\Tau} )
  = &  2 \ellc^{\prime} \otimes_s {\uwidehat{X}}^{\prime} + \Q^{\prime} \ellc^{\prime} \otimes \ellc^{\prime}
  - \frac{1}{\n-1} \left ( \Q^{\prime} \ll{}^{\prime} + 2 \ellc^{\prime} (X') \right )
  \gamma \nonumber \\
   = &  2 z^{2+r} \big ( \ellc + \uwidehat{\gauge} \big ) \otimes_s
  \big (  \uwidehat{X} - \Q \uwidehat{\gauge} \big )
  + z^{2+r} \Q \big ( \ellc + \uwidehat{\gauge} \big ) \otimes_s
  \big ( \ellc + \uwidehat{\gauge} \big ) \nonumber \\
 & - \frac{z^{2+r}}{\n-1} \Big ( \Q \big ( \ll + 2 \ellc(\gauge)
  + \uwidehat{\gauge} (\gauge) \big ) + 2 (\ellc +  \uwidehat{\gauge} ) ( X - \Q \gauge )
  \Big ) \gamma \nonumber \\
  = &
  z^{2+r} V^{(p)}_{\Tau}
  + z^{2+r} \left ( -\Q \uwidehat{\gauge} \otimes \uwidehat{\gauge}
  + 2 \uwidehat{X} \otimes_s \uwidehat{\gauge} 
+ \frac{1}{\n-1} \left (
\Q \uwidehat{\gauge} (\gauge)  - 2 \uwidehat{\gauge} (X) 
\right ) \gamma \right ).  \label{transVp}
\end{align}
It is immediate to check the second term in the last expression belongs to the kernel $\K_0$, so the tensor $V^{(p)}_{\Tau}$ does not have well-defined gauge weight, but the failure comes from the addition of a term in the kernel. This suggests that perhaps some other choice of $V_{\Tau} \in \tau^{-1} (\Tau)$ may have a well-defined weight. This turns out to be the case, but the analysis is different for $\Q=0$ and $\Q\neq 0$ (as noted above, being in one case or the other
is a gauge invariant statement).

\begin{lemma}
Let $\{\N,\gamma,\ellc,\ll \}$ be null metric hypersurface data
of dimension $\n \geq 2$. Assume that 
$\Tau \in \mbox{Ran} (\tau)$ has gauge weight $r$. Decompose $\Tau$ uniquely
according to (\ref{T(X,Q)}). Then,
\begin{itemize}
\item[(i)] If $\Q \neq 0$ then the tensor
\begin{align*}
V_{\Tau} \defi & 
2 \ellc \otimes_s \uwidehat{X} +  \Q \ellc \otimes  \ellc
+ \frac{1}{\Q} \uwidehat{X} \otimes \uwidehat{X} \\
& - \frac{1}{\n-1} \left ( \Q \ll + 2 \ellc (X) + \frac{1}{\Q} \uwidehat{X} (X) \right ) \gamma
\end{align*}
satisfies 
$\tau(V_\Tau) = \Tau$ and has gauge weight $r+2$.
\item[(ii)]  If $\Q=0$ assume that there exists a covector $\bmomega \in T^{\star}_p \N$ with
gauge behaviour $\G_{(z,\gauge)} (\bmomega) = z ( \bmomega + \uwidehat{\gauge})$ Then, the tensor
\begin{eqnarray*}
V_\Tau  \defi
2 ( \ellc - \bmomega ) \otimes_s \uwidehat{X}  - \frac{1}{\n-1} \left ( 2 (\ellc - \bmomega) (X)\right ) \gamma,
\end{eqnarray*}
satisfies
$\tau(V_\Tau) = \Tau$ and has gauge weight $r+2$.
\end{itemize}
\end{lemma}

\begin{proof}
 In the case $\Q \neq 0$, the difference tensor
\begin{eqnarray*}
V^{(p)}_{\Tau} - V_{\Tau} = \frac{1}{Q} \left ( - \uwidehat{X} \otimes \uwidehat{X} + \frac{1}{\n-1} \uwidehat{X} (X)  \gamma \right ) \defi \mbox{Diff}
\end{eqnarray*}
is easily seen to belong to the kernel 
$\K_0$, so $\tau(V_\Tau) = \Tau$ follows. For the gauge behaviour,
we note that the gauge group acts on $\mbox{Diff}$ as follows
\begin{align*}
  \G_{(z,\gauge)} (\mbox{Diff}) & =
  \frac{z^{r+2}}{\Q} \left ( - (\uwidehat{X} - \Q \uwidehat{\gauge})
  \otimes_s (\uwidehat{X} - \Q \uwidehat{\gauge})
  + \frac{1}{\n-1} (\uwidehat{X} - \Q \uwidehat{\gauge}) ( X - \Q \gauge) \gamma \right ) \\
  & =
  z^{r+2} \mbox{Diff}  + z^{r+2} \left ( - \Q \uwidehat{\gauge} \otimes \uwidehat{\gauge}
+ 2 \uwidehat{X} \otimes_s \uwidehat{\gauge} 
  + \frac{1}{\n-1} \left ( \Q \uwidehat{\gauge} (\gauge)  - 2 \uwidehat{\gauge} (X)  \right ) \gamma \right ),
\end{align*}
which combined with (\ref{transVp}) implies the gauge covariance
$\G_{(z,\gauge)} (V_{\Tau} )  = z^{r+2} V_{\Tau}$.

The case $\Q=0$ is immediate because,
on the one hand the difference tensor, given by 
\begin{align*}
  V^{(p)}_{\Tau} - V_{\Tau} = 2 \bmomega \otimes_s \uwidehat{X} -  \frac{1}{\n-1} \left (
  2 \bmomega(X) \right ) \gamma,
\end{align*}
clearly   belongs to the kernel and, on the other hand,
$X$ has gauge weight $r+1$ (by Lemma  \ref{GaugeXQ} and $\Q=0$) 
and $\ellc - \bmomega$ has gauge weight $+1$ (by the second in
\eqref{gaugemhd}), so $V_{\Tau}$ has gauge weight $r+2$.
\end{proof}

\section{The energy-momentum map: PDE properties}
\label{EMmap2}
All the results of the previous section are purely algebraic. In this section we want to exploit the decomposition of symmetric tensors induced by the energy-momentum map to define  differential operators that respect the structure.

We start by writing down  the Lie derivative along $\upn$
of various tensors on $\N$.
\begin{lemma}
Let $\{\N,\gamma,\ellc,\ll\}$ be null metric hypersurface data. Then
\begin{align}
&\pounds_{\upn} \gamma_{ab}  = 2 \U_{ab} \label{poundsgamma}\\
&\pounds_{\upn} \ell_{a}  =  2 \sone_a  \label{poundsell} \\
&\pounds_{\upn} P^{ab}  = - 2 \sone_c  \left ( P^{ac} \vecn^b
+ P^{bc} \vecn^a \right ) - 2 P^{ac} P^{bf} \U_{cf} - \vecn^a \vecn^b \vecn^c \nablao_{c} \ll
\label{poundsP}
\end{align} 
\end{lemma}
\begin{proof}
  \eqref{poundsgamma} is just the definition of the tensor $\bU$ in \eqref{defU}. (\ref{poundsell}) follows from the fact that
$\nablao$ has vanishing torsion so that Lie derivatives can be computed using covariant derivatives. This together with 
(\ref{nablaoll}) and (\ref{nablaon}) implies 
\begin{align*}
\pounds_{\upn} \ell_a = \vecn^b \nablao_b \ell_a + \ell_b
\nablao_a \vecn^b = \F_{ba} \vecn^b + \ell_b \left (  \vecn^b \sone_a  +
P^{bf} \U_{af} \right ) =  2 \sone_a  - \ll \vecn^f \U_{af} =  2 \sone_a.
\end{align*}
Finally (\ref{poundsP}) follows from (\ref{nablaon}) and (\ref{nablaoP}), which gives
\begin{align*}
\pounds_{\upn} P^{ab} & = \vecn^c \nablao_c P^{ab} - P^{cb} \nablao_{c} \vecn^a
- P^{ac} \nablao_{c} \vecn^b \\
& = 
- \left ( \vecn^a P^{bf} + \vecn^b P^{af} \right ) \F_{cf} \vecn^c
- \vecn^a \vecn^b \vecn^c \nablao_c \ll 
+ P^{cb} \left ( \F_{cf} \vecn^f \vecn^a - P^{af}  \U_{cf} \right )
+ P^{ca} \left ( \F_{cf} \vecn^f \vecn^b - P^{bf}  \U_{cf} \right ) 
\end{align*}
which becomes (\ref{poundsP}) after rearrangement.
\end{proof}

At this point we introduce a  linear map $\LL$ acting
on symmetric two-co\-va\-riant tensors. This map allows us to define
a natural differential operator taking  values in $\K_0$ when acting on tensors in $\K_0$.

\begin{proposition}
\label{hatLn}
Consider null metric hypersurface data  $\metdata$ 
and let $\LL$ be the linear map $\LL : S^0_2 \longrightarrow S^0_2$ defined by
\begin{equation}
\LL(V)_{ab} \defi - 2 P^{cd} V_{c(a} \U_{b)d} 
- 4 \uwidehat{\X}_{(a} \sone_{b)} - \frac{\Vnn}{\n-1} \upn (\ll) \gamma_{ab}.
\label{defL}
\end{equation}
Define the differential operator acting on symmetric $(0,2)$-tensors
\begin{align}
\hatLn (V) := \pounds_{\upn} V + \LL(V). \label{defhatLn}
\end{align}
Then the following properties holds
\begin{align*}
& (i) \qquad  \quad  \left [ \hatLn, \P \right ]  = 0, \\
& (ii) \qquad  \quad \trP (\hatLn (V)) = \pounds_{\upn} \left ( \trP V \right ).
\end{align*}
\end{proposition}

\begin{proof}
We start with property (i). Let us compute the commutator 
\begin{align*}
\left [ \hatLn, \P \right ] (V) & = \hatLn (\P(V)) - \P ( \hatLn (V)) \\
& = \pounds_{\upn} ( \P (V) ) + \LL ( \P (V)) - \P ( \pounds_{\upn} (V) ) - \P ( \LL (V) ) \\
& = \left ( \pounds_{\upn} \P + [ \LL, \P ]  \right ) (V)
\end{align*}
where in the first term of the last expression
we are viewing $\P$ as the  $(2,2)$ tensor field on $\N$ defined in
\eqref{defTensorP}. Thus, we need
to prove that $\pounds_{\upn} \P + [ \LL, \P ] =0$. We start with
$\pounds_{\upn} \P$. Using \eqref{poundsgamma}-\eqref{poundsP} and the
obvious property $\pounds_{\upn} \upn =0$ it follows
\begin{align*}
\left ( \pounds_{\upn} \P \right )_{ab}^{\,\,\,\,\,\,cd}  = & 
2 \delta^{(c}_{(b} \vecn^{d)} \pounds_{\upn} \ell_{a)} - 2 \vecn^c \vecn^d
\ell_{(a} \pounds_{\upn} \ell_{b)}
+ \frac{1}{\n-1} \left ( \ll \vecn^c \vecn^d + P^{cd} \right ) \pounds_{\upn} \gamma_{ab} \\
& + \frac{1}{\n-1} \left (  \upn (\ll) \vecn^c \vecn^d + \pounds_{\upn} P^{cd}
\right ) \gamma_{ab} \\
= & 
4  \delta^{(c}_{(b} \vecn^{d)} \sone_{a)} - 4 \vecn^c \vecn^d \ell_{(a} \sone_{b)}
+ \frac{2}{\n-1} \left ( \ll \vecn^c \vecn^d + P^{cd} \right ) \U_{ab} \\
& + \frac{1}{\n-1} \left ( - 4 \sone_{e} P^{e(c} \vecn^{d)}
- 2 P^{ce} P^{df} \U_{ef} \right ) \gamma_{ab}.
\end{align*}
Contracting with  $V \in \Gamma(S^0_2)$ this gives
gives (recall Terminology \ref{TerVs})
\begin{align}
\left ( \pounds_{\upn} \P \right )_{ab}^{\,\,\,\,\,\,cd}  V_{cd}  = &  
 \, 4 \delta^c_{(a} \vecn^d \sone_{b)} V_{cd}
- 4 \Vnn \ell_{(a} \sone_{b)} 
+ \frac{2}{\n-1} \left (  \ll \Vnn + \Z \right ) \U_{ab} \nonumber \\
& + \frac{1}{\n-1} \left ( - 4 \sone_{e}  P^{ec} \vecn^d V_{cd} - 2
P^{ce} P^{df} V_{cd} \U_{ef} \right ) \gamma_{ab} \nonumber \\
= & \, 
\left (  4 \delta^c_{(a} \sone_{b)}  
- \frac{4}{\n-1} \sone_{e}  P^{ec} \right )
\left ( \uwidehat{\X}_c +  \Vnn \ell_c \right )
- 4 \Vnn \ell_{(a} \sone_{b)}  \nonumber \\
& + \frac{2}{\n-1} \left [
\left (  \ll \Vnn + \Z \right ) \U_{ab}  - 
P^{ce} P^{df} V_{cd} \U_{ef}  \gamma_{ab} \right ]\nonumber \\
 = & \, 
  4 \uwidehat{\X}_{(a} \sone_{b)}  
- \frac{4}{\n-1} \sone_{e}  P^{ec} \uwidehat{\X}_c \nonumber \\
& + \frac{2}{\n-1} \left [
\left (  \ll \Vnn + \Z \right ) \U_{ab}  - 
P^{ce} P^{df} V_{cd} \U_{ef}  \gamma_{ab} \right ], \label{c1}
\end{align}
where in the last equality we used $P^{ec} \ell_c = - \ll \vecn^e$
and $\bsone(\upn)=0$. We next compute the
commutator $[\LL,\P]$. It is obvious from the definition of $\LL$ that
\begin{align}
  \label{PVh}
  \LL(\VH)_{ab} = -2 P^{cd} \VH_{c(a} \U_{b)d}, \qquad  \mbox{for} \quad \VH \in \K_0.
\end{align}
Thus,
\begin{align}
\LL (\P(V))_{ab} & = \LL (V-\VH)_{ab} = \LL(V)_{ab} - \LL(\VH)_{ab} \nonumber \\
& =  
- 2 P^{cd} \left ( V_{c(a}- \VH_{c(a} \right ) \U_{b)d} 
- 4 \uwidehat{\X}_{(a} \sone_{b)} - \frac{\Vnn}{\n-1}  \upn (\ll) \gamma_{ab}  \nonumber \\
& =
- 2 P^{cd} \P(V)_{c(a} \U_{b)d} 
- 4 \uwidehat{\X}_{(a} \sone_{b)} - \frac{\Vnn}{\n-1}  \upn (\ll) \gamma_{ab}. 
\label{c2}
\end{align}
To compute $\P (\LL(V))$ we note that $\uwidehat{\X}$ and $\bsone$ are both orthogonal
to $n$. Hence, the definition of $\LL(V)$  and properties (ii), (iv) in Remark \ref{propertiesProjector} give
\begin{align}
\P(\LL(V))_{ab} = -2 \P_{ab}^{\,\,\,\,\,ef} P^{cd} V_{c(e} \U_{f)d}
- \frac{1}{\n-1}  \left ( 4 P^{cd} \sone_{c}  \uwidehat{\X}_d
+ \Vnn \upn(\ll) \right ) \gamma_{ab}. \label{c3}
\end{align}
Combining (\ref{c1}), (\ref{c2})  and (\ref{c3}) yields
\begin{align}
\left ( \pounds_{\upn} \P + [ \LL,\P]  \right ) (V)_{ab} = &
\frac{2}{\n-1} \left ( \left (  \ll \Vnn + \Z \right ) \U_{ab}
- P^{ce} P^{df} V_{cd} \U_{ef} \gamma_{ab} \right ) \nonumber \\
& - 2 P^{cd} \P(V)_{c(a} \U_{b)d} + 2 \P_{ab}^{\,\,\,\,\,\,ef} P^{cd} V_{c(e} \U_{f)d}. \label{almost}
\end{align}
We need to elaborate the last two terms. For the last one we define $W_{ef} \defi 2 P^{cd} V_{c(e} \U_{f)d}$
and note that $W_{ef} \vecn^f = P^{cd} V_{cf} \U_{ed} \vecn^f = 
P^{cd} (\uwidehat{\X}_c + \Vnn \ell_c) \U_{ed} = P^{cd} \uwidehat{\X}_c \U_{ed}$
after using $P^{cd} \ell_{c} = - \ll \vecn^d$ and $\U_{ed} \vecn^d =0$. Hence,
$W_{ef} \vecn^e \vecn^f =0$. In addition 
$P^{ef} W_{ef} = 2 P^{cd} P^{ef} V_{ce} \U_{fd}$.  The definition of $\P$ yields
\begin{align}
2 \P_{ab}^{\,\,\,\,\,ef} P^{cd} V_{c(e} \U_{f)d} & =
\P (W)_{ab}  \nonumber \\
& = \ell_a P^{cd} \uwidehat{\X}_c \U_{bd} +
\ell_b P^{cd} \uwidehat{\X}_c \U_{ad} + \frac{2}{\n-1} P^{cd} P^{ef} V_{ce} \U_{fd} 
\gamma_{ab}. \label{last}
\end{align}
For the penultimate term in (\ref{almost}), the definition of
$\P(V)$ and the facts that $P^{cd} \ell_c \U_{db} = - \ll \vecn^d \U_{db} =0$
and
$P^{cd} \gamma_{ca} \U_{bd}  = \left ( \delta^d_a - \vecn^d \ell_a \right ) \U_{db}
= \U_{ab}$
  imply
\be
P^{cd} \P(V)_{ca} \U_{db} = 
\ell_a P^{cd} \uwidehat{\X}_c \U_{db} + \frac{1}{\n-1} \left ( \Vnn \ll + \Z \right ) \U_{ab}
\en
so that
\begin{equation}
-2 P^{cd} \P(V)_{c(a} \U_{b)d} = 
- \ell_a P^{cd} \uwidehat{\X}_c \U_{db} - \ell_b P^{cd}  \uwidehat{\X}_c \U_{ad} 
- \frac{2}{\n-1} \left ( \Vnn \ll + \Z \right ) \U_{ab}.
\label{last_but_one}
\end{equation}
Inserting (\ref{last}) and (\ref{last_but_one}) into (\ref{almost}) proves
$( \pounds_{\upn} \P + [ \LL,\P] ) (V)=0$ and property (i) follows.

To prove property (ii) we first observe that the definition of $\LL$ yields
\begin{align}
\trP \LL(V) & = \LL(V)_{ab} P^{ab} = - 2 P^{cd} V_{ca} \U_{bd} P^{ab}
- 4 \uwidehat{\X}_a \sone_{b}  P^{ab} -  \Vnn \upn(\ll).  \label{first}
\end{align}
Recalling (\ref{poundsP}) it follows
\begin{align}
\trP \left ( \pounds_{\upn} V\right ) & = P^{ab} \pounds_{\upn} V_{ab} =
\pounds_{\upn} \left ( V_{ab} P^{ab} \right ) - V_{ab}
\pounds_{\upn} P^{ab} \nonumber \\
& = 
\pounds_{\upn} (\trP V)  - V_{ab} \left ( - 4 \sone_{c} P^{ca} \vecn^{b}
- 2 P^{ac} P^{bf} \U_{cf} - \vecn^a \vecn^b \upn (\ll) \right ) \nonumber \\
& = 
\pounds_{\upn} (\trP V)  + 4 \sone_c P^{ca} \left ( \uwidehat{\X}_a + \Vnn \ell_a \right )
+ 2 V_{ab} P^{ac} P^{bf} \U_{cf} +  \Vnn \upn (\ll) \nonumber \\
& = 
\pounds_{\upn} (\trP V)  + 4 \sone_{c} P^{ca} \uwidehat{\X}_a 
+ 2 V_{ab} P^{ac} P^{bf} \U_{cf} + \Vnn \upn (\ll). \label{second}
\end{align}
where in the last equality we used $\sone_{c}  P^{ca} \ell_a = - 
\ll \F_{fc} \vecn^f \vecn^c = 0$. Adding (\ref{first}) and (\ref{second}):
\begin{align*}
\trP \left ( \hatLn(V) \right ) = 
\trP \left ( \pounds_{\upn} (V) \right )
+ \trP \LL(V) = \pounds_{\upn} \left ( \trP V \right ).
\end{align*}
\end{proof}

\begin{corollary}
If $\VH \in \K_0$ then $\hatLn(\VH) \in \K_0$ and
\be
\hatLn(\VH)_{ab} = 
\pounds_{\upn} \VH_{ab}  - P^{cd} \left( \VH_{ac} \U_{bd} + 
\VH_{bc} \U_{ad} \right ).
\en
\end{corollary}
\begin{proof}
By Corollary \ref{I-P} the map $\I - \P$ projects $S^0_2$ into the kernel $\K_0$. Since
\be
\hatLn ( (\I- \P)(V)) = \hatLn(V) - \hatLn(\P(V)) = 
\hatLn(V) - \P ( \hatLn(V)) = \left ( \I - \P \right ) 
\hatLn(V) 
\en
it follows that $\hatLn(\VH) \in \K_0$. For the second statement
we simply use  \eqref{PVh}.
\end{proof}

Proposition \ref{hatLn} allows us to obtain easily the explicit expression
of $\hatLn( \P(V))$.
\begin{lemma}
\label{hatLnP(V)Lemma}
Let $V$ be a symmetric $(0,2)$-tensor field in a null metric hypersurface
data $\metdata$.
Then
\begin{align}
\hatLn (\P (V))_{ab} = & 2 
\ell_{(a} \left ( \pounds_{\upn} \uwidehat{\X}_{b)} 
                         + 2 \Q \sone_{b)}  - P^{cd} \U_{b)d} 
                         \uwidehat{\X}_c \right ) 
+ \upn(\Q) \ell_a \ell_b + \frac{1}{\n-1} \left ( \upn (\Q) \ll + \upn(\Z)
\right )  \gamma_{ab}.
\label{hatLnP(V)}
\end{align}
\end{lemma}

\begin{proof}
Since $\hatLn(\P(V)) = \P(\hatLn(V))$ the result will follow as an application 
of Proposition \ref{defProjector} with $V_{ab}$ replaced by $\hatLn(V)_{ab}$. Thus, we need
to compute
\begin{align*}
(\hatLn(V))_{ab} \vecn^b & = 
(\pounds_{\upn} V_{ab}) \vecn^b + \LL(V)_{ab} \vecn^b \\
& = 
\pounds_{\upn} \left ( V_{ab} \vecn^b \right ) 
- P^{cd} V_{cb} U_{ad} \vecn^b \\
& = \pounds_{\upn} \left ( \uwidehat{\X}_a + \Q \ell_a \right )
- P^{cd} U_{ad} \left ( \uwidehat{\X}_c + \Q \ell_c \right ) \\
& = 
\pounds_{\upn} \uwidehat{\X}_a + \upn(\Q) \ell_a + 2 \Q \sone_{a} 
- P^{cd} U_{ad} \uwidehat{\X}_c,
\end{align*}
where in the second equality we used
$\LL(V)_{ab} \vecn^b = - P^{cd} V_{cb} U_{ad} \vecn^b$ which follows at once from 
the definition (\ref{defL}) of $\LL(V)$.
Thus
\be
\hatLn(V)_{ab} \vecn^a \vecn^b =  \upn (\Q)
\en
and (recall \eqref{unX})
\be
\uwidehat{\overline{\hatLn(V)}}_a =
    \hatLn(V)_{ab} \vecn^b - (\vecn^b \vecn^c \hatLn(V)_{bc}) \ell_a = 
\pounds_{\upn} \uwidehat{\X}_a + 2 \Q \sone_a 
- P^{cd}  U_{ad} \uwidehat{\X}_c. 
\en
Moreover, property (ii) of Proposition \ref{hatLn} gives
\be
\hatLn(V)_{ab} P^{ab} = \pounds_{\upn} (V_{ab} P^{ab} ) = \upn (\Z).
\en
Inserting into (\ref{Proj}) with $V \longrightarrow \hatLn(V)$ 
gives (\ref{hatLnP(V)}) at once.
\end{proof}

We have already shown that $\hatLn$ is a well-defined map on the space
of symmetric tensor fields in the kernel of
$\tau$ (by item (i) in Proposition \ref{hatLn}). In the next lemma we introduce another differential operator that also takes values on the kernel of $\tau$. In this case, the operator acts on covectors orthogonal to $\upn$.

\begin{lemma}
\label{defO}
Consider $\metdata$ null metric hypersurface data. Let
$\bmomega\in \XX^{\star} (\N)$ be orthogonal to $\upn$, i.e.
$\bmomega(n)=0$. 
Then the derivative operator
\begin{align}
\O(\bmomega)_{ab} := \nablao_a \omega_b + \nablao_b \omega_a - 2 \ell_{(a} 
\left ( \pounds_{\upn} \omega_{b)} - 2 \U_{b)f} P^{fc} \omega_c \right )
- \frac{2}{\n-1} P^{cf} \nablao_c \omega_f \gamma_{ab} \label{defOexp}
\end{align}
takes values in the kernel of the  energy-momentum map.
\end{lemma}

\begin{proof}
We want to apply the general decomposition (\ref{Form3}) to the tensor
$T_{ab} \defi \nablao_a \omega_b + \nablao_b \omega_a$. The contraction with $\vecn^b$ gives
\begin{align*}
 T_{ab} \vecn^b & = \left ( \nablao_a \omega_b + \nablao_b \omega_a \right ) \vecn^b \\
& = - \omega_b \nablao_a \vecn^b + \vecn^b \nablao_b \omega_a = \pounds_{\upn} \omega_a - 
2 \omega_b \nablao_a \vecn^b  \\
& = \pounds_{\upn} \omega_a  - 2 \omega_b \left (  \vecn^b \sone_a + 
P^{bf} \U_{af} \right ) \\
& = \pounds_{\upn} \omega_a  - 2 \omega_b P^{bf} \U_{af},
\end{align*}
where we used $\bmomega (\upn) =0$ and 
(\ref{nablaon}) has been inserted in the fourth step.
Consequently $T_{ab} \vecn^a \vecn^b=0$
and the decomposition (\ref{Form3}) gives
\begin{align*}
\nablao_a \omega_b + \nablao_b \omega_a = & 
2\ell_{(a} \left ( \pounds_{\upn} \omega_{b)} - 2 \U_{b)f} P^{fc} \omega_c \right ) 
+ \frac{1}{\n-1} P^{cf} \left ( \nablao_c \omega_f 
                                      +\nablao_f \omega_c \right ) \gamma_{ab}+ \O^H(\bmomega)_{ab}
\end{align*}
where $\O^H$ lies in the kernel of $\tau$. 
Since by definition $\O(\bmomega)= \O^H$ the result follows.
\end{proof}

\section{The constraint tensor: decomposition in irreducible components}
\label{ConstTensorDecom}

So far all the results have only required a null metric hypersurface data set. 
This geometric notion encodes at the abstract level the ``intrinsic'' geometry of a null hypersurface, by which we mean the metric of the ambient space evaluated at the points on the hypersurface. The connection is made via the notion of 
embedded metric hypersurface data \cite{MarsGRG, Mars2020}.

\begin{definition}[Embedded null metric hypersurface data]
  \label{embed1}
  Let $\metdata$ be null metric hypersurface data  of dimension $\n$ and $(M,g)$ be a pseudo-riemannian manifold of dimension $\n+1$. We say that $\metdata$ is embedded in $(M,g)$ with embedding
$\Phi$ and rigging $\rig$ if there exists an embedding $\Phi: \N \longrightarrow 
M$ and a vector field $\rig$ along $\Phi(\N)$ everywhere transversal to
$\Phi(\N)$ (i.e. a so-called {\bf rigging vector}) such that
\begin{align*}
\Phi^{\star} (g)= \gamma, \qquad
\Phi^{\star} ( g (\rig,\cdot) ) = \ellc, \qquad
\Phi^{\star} (g(\rig,\rig)) = \ll.
\end{align*}
\end{definition}
An immediate consequence of the definition is that the signature of $g$ must be the same as the signature of the tensor $\A$ defined in \eqref{defA}. In fact, in the embedded case both tensors are equivalent. It is in this sense that 
the null metric hypersurface data encodes the ambient metric at an abstract level.  Since the tensors $\{n, P\}$ arise as elements in the decomposition of the
contravariant tensor $\A^{\sharp}$, the contravariant metric $g^{\sharp}$ can  be expressed in terms of $\upn$ and $P$. Let $\{e_a\}$ be a (local) basis of $T \N$
so that $\{\rig, \Phi_{\star}(e_a)\}$ is a (local basis) of $T M$ along the hypersurface
$\Phi(\N)$. Then one can express (see \cite{Gabriel1} for details)
  \begin{align}
    \label{contra}
    g^{\sharp} \stackrel{\Phi(\N)}{=}
    2 \rig \otimes_s \nu+ P^{ab} \Phi_{\star} (e_a) \otimes \Phi_{\star} (e_b)
  \end{align}
  where $\nu := \Phi_{\star}(\upn)$  is the null normal to $\Phi(\N)$ satisfying
  $g(\nu,\rig)=1$. It is easy to show \cite{Mars2020} that if $\metdata$ is embedded in an 
ambient space with embedding $\Phi$ and rigging $\rig$, then for all $(z,\gauge)
\in \G$, the gauge transformed data
$\{ \N, \G_{(z,\gauge)}(\gamma),
\G_{(z,\gauge)}(\ellc),
\G_{(z,\gauge)}(\ll)\}$ is also embedded in $(M,g)$ with the same embedding and rigging
\begin{align}
  \G_{(z,\gauge)} (\rig) := z \left ( \rig + \Phi_{\star} (\gauge) \right ).
\label{gaugerig}
\end{align}
Now, null hypersurfaces in an ambient space possess, in addition to
an ``intrinsic'' geometry (in the sense above) also an ``extrinsic'' one in the sense of encoding 
first transversal derivatives of the ambient metric  at the hypersurface. 
At the abstract level this leads \cite{MarsGRG, Mars2020} to the definition 
of {\bf null  hypersurface data}, which we recall next, together with its appropriate notion of {\em embeddeness}.
\begin{definition}
A {\bf null hypersurface data} is a $5$-tuple $\hypdata$ 
where $\metdata$ is null metric hypersurface data and $\bY$ is a smooth
symmetric $(0,2)$-tensor field on $\N$ on which the gauge group acts as
\begin{align}
  \G_{(z,\gauge)} \bY = z \bY + \ellc \otimes_s dz
  + \frac{1}{2} \pounds_{z\gauge} \gamma. \label{gaugebY}
\end{align}
We say that $\hypdata$ of dimension $\n$ is embedded in a 
pseudo-riemannian manifold $(M,g)$ of dimension $\n+1$ with embedding $\Phi$ and
rigging $\rig$, provided the metric part of the data
$\metdata$ is embedded in $(\M,g)$ in the sense of Definition \ref{embed1}
and, moreover,
\begin{align*}
\frac{1}{2} \Phi^{\star} \left  ( \pounds_{\rig} g \right ) = \bY.
\end{align*}
\end{definition}
As before, if $\hypdata$ is embedded with embedding $\Phi$ and
rigging $\rig$, the gauge transformed data with gauge parameters $(z,\gauge)$
is also embedded with the same embedding and rigging given by \eqref{gaugerig}.

The connection $\nablao$  defined intrinsically on any null metric hypersurface data (and hence also in any  hypersurface data) can be related, in the embedded case, to the Levi-Civita covariant derivative of $(M,g)$ along tangential directions of $\Phi(\N)$. The specific result is
\cite{MarsGRG, Mars2020}

\begin{lemma}
Let $\hypdata$ be null hypersurface data embedded in $(M,g)$ with embedding $\Phi$ and rigging $\rig$. Let $\nabla$ be the Levi-Civita derivative of $(M,g)$. Then, for any two vector fields $X,W \in \XX(\N)$ it holds
\begin{align*}
  \Phi_{\star} ( \nablao_X W ) = \nabla_{\Phi_{\star} (X)} \Phi_{\star} (W)
  - \Phi_{\star} ( \bY(X,W)) \nu - \Phi_{\star} (\bU(X,W))  \rig
 \end{align*} 
where $\nu$ is the unique normal vector to $\Phi(\N)$ satisfying 
$g (\nu,\rig)=1$
and for any function $f \in \FF(\N)$ we define $\Phi_{\star} (f): \Phi(\N)
\rightarrow \mathbb{R}$ by $\Phi_{\star}(f) \circ \Phi = f$.
\end{lemma}
In a  null hypersurface data set one can define a geometrically relevant symmetric
$(0,2)$-tensor called {\bf constraint tensor}. This tensor is defined at the level of curvature and involves both the Ricci tensor of the connection $\nablao$ as well as covariant derivatives or products of the rest of the hypersurface data terms.
Our convention for the curvature tensor $\Riem^{D}$ of a connection $D$ is
\begin{align*}
  \Riem^D (\bm{\omega}, X,Z,W) := \bm{\omega} \left (
    D_Z D_W X - D_W D_Z X - D_{[Z,W]} Z \right )
\end{align*}
where ${\bf \omega}$ is a covector and $X,Z,W$ vector fields and $[Z,W] := \pounds_ZW$ is the Lie bracket of two vector fields. The Ricci tensor of $D$, denoted by
$\Ricc^D$, is the trace of $\Riem^D$ in the first and third indices. For the connection $\nablao$ on $\N$ we shall use
$\Riemo$ and $\Ricco$ respectively. Note that $\Ricco$ is not a symmetric tensor in general. For the Levi-Civita connection $\nabla$
in $(M,g)$ we use $\Riem_g$ and $\Ricc_g$.

For notational simplicity we shall give special names to $\iota_n \bY$ and
$\Q_{\bY}$, so we define (note the relative sign between $\kappa_n$ and
$\Q_{\bY}$)
\begin{align*}
\textcolor{black}{\brone := \iota_n \bY, \qquad \kappa_n := - \Q_{\bY} = - \bY(\upn,\upn).}
\end{align*}
The constraint tensor has been defined in \cite{Gabriel1} in the context of
so-called characteristic hypersurface data, which is a particular case of
null  hypersurface data, and in \cite{ManzanoMarsConstraint} for general hypersurface data. In the null case the definition is as follows
\begin{definition}[Constraint tensor]
  Let $\hypdata$ be null hypersurface data. The constraint tensor, denoted by
  $\R$ is the symmetric $(0,2)$-tensor field
      \textcolor{black}{
  \begin{align}    
      \R_{ab} = & \Ricco_{(ab)} - 2
    \pounds_{\upn} \Y_{ab}
    - \left ( 2\kappa_n+ \mbox{tr}_P \bU \right ) \Y_{ab}
    + \nablao_{(a} \left ( \sone_{b)}  + 2 \rone_{b)} \right ) \nonumber \\
    & - 2 \rone_a \rone_b + 4 \rone_{(a} \sone_{b)}
    - \sone_a \sone_b
    - (\mbox{tr}_P \bY) \U_{ab}
    + 2 P^{cd} \U_{d(a} \left ( 2 \Y_{b)c} + \F_{b)c} \right ).
\label{ConstraintTensor}
  \end{align}}
  \end{definition}
The  importance of the constraint tensor lies in the fact that
it encodes, at the abstract level, the information of the Ricci tensor along tangential directions of any embedded null hypersurface. Specifically we have the following result
  \cite{Gabriel1, ManzanoMarsConstraint}.
  \begin{proposition}
    Let $\hypdata$ be null hypersurface data embedded in $(M,g)$ with embedding
    $\Phi$ and rigging $\rig$. Then the following identity holds
    \begin{align*}
      \Phi^{\star} ( \Ricc_g) = \R.
    \end{align*}
  \end{proposition}
  The constraint tensor is a symmetric $(0,2)$-tensor, so all the results derived in the previous sections can be applied. In particular we intend to rewrite the tensor using the decomposition of $\bY$ according to Proposition \ref{Decom}
  and using the differential operators introduced in the previous section.

By Proposition \ref{Decom}, the tensor $\bY$ can be decomposed  as
\textcolor{black}{
\begin{align}
\Y_{ab} = 2 \ellc_{(a} \XY_{b)} 
- \kappa_n \left ( \ell_a \ell_b + \frac{1}{\n-1} \ll \gamma_{ab} \right )
+ 
\frac{\ZY}{\n-1}  \gamma_{ab} 
+ \YH_{ab}, \label{decomY}
\end{align}}
where the covector $\XY$ (c.f. Terminology
  \ref{TerVs}) is related to $\brone$ by means of (c.f. \eqref{unX})
  \textcolor{black}{ \begin{align}
    \brone = \XY - \kappa_n \ellc.
\label{defXY}
  \end{align}}
The tensor $\bU$ being symmetric can also be decomposed uniquely according to Proposition \ref{Decom}. Since $\bU(n,\cdot)=0$, the only terms that survive are the trace term and the component in the kernel, see Remark \ref{orthogonalcase}
\begin{align}
 \bU = \frac{\mbox{tr}_P \bU}{\n-1}  \gamma + \bUH.
\label{decombU}
\end{align}

In the decomposition of Proposition \ref{Decom} all covectors involved
are orthogonal to $\upn$ except for $\ellc$. This suggests the convenience of introducing a derivative on scalar functions $f \in \FF(\N)$ that respects this property. So, we define $D f$ by means
$D f := d f - \upn(f) \ellc$ so that the 
contraction with $\upn$ gives $\iota_{\upn}  D f =0$. Thus, we will replace $\nablao$ derivatives acting on scalars in terms
of $D$ derivatives by means of
\begin{align}
\nablao f= D f + \upn(f) \ellc. \label{nablaD}
\end{align}

\begin{theorem}
  \label{decomConstThm}
  Let $\metdata$ by null hypersurface data. Then the constraint tensor
  $\R$ admits the following decomposition
\textcolor{black}{
  \begin{align}
 \R_{ab} = & \Ricco_{(ab)}  + \nablao_{(a} \sone_{b)} + \sone_a \sone_b
+ 2 P^{cd} \U_{d(a} \F_{b)c} \nonumber \\
& + \ell_{(a} \left ( -2 \pounds_\upn   \XY_{b)} - 2 D_{b)} \kappa_n +
4 \kappa_n \sone_{b)}
                - 2 (\ZU) \XY_{b)}  \right ) 
+ \kappa_n (\ZU)  \tgamma_{ab} 
\nonumber \\
& + \frac{2}{\n-1} \left (  \upn( \kappa_n \ll - \ZY)  
+ \trP \left ( \nablao \,\, \bm{\XY} - (\bm{\XY}+  \bsone) \otimes (\bm{\XY} + \bsone) \right )
+ (\kappa_n \ll - \ZY) ( \kappa_n + \ZU)
 \right ) \gamma_{ab} \nonumber \\
        & - 2 \hatLn \YH_{ab} +  \O(\bm{\XY})_{ab} - 2 \left ( (\bm{\XY} + \bsone)
        \otimes (\bm{\XY} + \bsone) \right ){}^{\!\!H}_{ab}
        - \left ( 2 \kappa_n + \ZU \right )
        \YH_{ab} + \left ( 2 \ll \kappa_n -\ZY 
        \right )  \UH_{ab}.  \label{decomConst}
  \end{align}}
\end{theorem}
\begin{proof}
  We start by rewriting the combination $-2 \rone_a \rone_b + 4 \rone_{(a} \sone_{b)}$ that appears in the constraint tensor in terms of $\XY$, c.f. \eqref{defXY},
\textcolor{black}{
  \begin{align*}
    -2 \rone_a \rone_b + 4 \rone_{(a} \sone_{b)}
    =  
    4 \kappa_n \ell_{(a} \left ( \XY_{b)} - \sone_{b)}  \right )
    - 2 \kappa_n^2 \ell_a \ell_b - 2 \XY_{a} \XY_b + 4 \XY_{(a} \sone_{b)}.
  \end{align*}}
  Using this and replacing $\pounds_\upn \bY$ in terms of
  the operator $\hatLn \bY$ introduced  in Proposition \ref{hatLn}
brings the constraint tensor \eqref{ConstraintTensor} into the form
\textcolor{black}{  \begin{align*}
    \R_{ab} = & \Ricco_{(ab)} - 2
                \hatLn \Y_{ab}
                + \frac{2 \kappa_n}{\n-1} \upn(\ll) \gamma_{ab}
    - \left ( 2\kappa_n+ \mbox{tr}_P \bU \right ) \Y_{ab}
    + \nablao_{(a} \left ( \sone_{b)}  + 2 \XY_{b)} - 2 \kappa_n \ell_{b)}  \right ) \\
              & - 2 \XY_a \XY_b + 4 \kappa_n \ell_{(a} \left ( \XY_{b)} - \sone_{b)}
        \right )
                - 2 \kappa_n^2 \ell_a \ell_b - 4 \XY_{(a} \sone_{b)}
    - \sone_a \sone_b
    - (\mbox{tr}_P \bY) \U_{ab}
    + 2 P^{cd} \U_{d(a} \F_{b)c}. 
  \end{align*}}
  We can replace $2\nablao_{(a} \XY_{b)}$ in terms of the derivative operator
  $\O$ defined in Lemma \ref{defO}, namely
  \begin{align}
    2 \nablao_{(a} \XY_{b)} = 
    \O(\bm{\XY})_{ab}
   +  \ell_{(a} 
\left ( 2 \pounds_{\upn} \XY_{b)} - 4 \U_{b)f} P^{fc} \XY_c \right )
    + \frac{2}{\n-1} (P^{cf} \nablao_c \XY_f) \gamma_{ab}
    \label{t1}
\end{align}
  and use also the fact that $\hatLn$ has good properties with respect to the projector $\P$, which gives
  \begin{align*}
    2\hatLn \Y_{ab} &= 2 \hatLn \left ( \P(\bY) + \bYH \right )_{ab} \\
                       & = 
\ell_{(a} \left ( 4 \pounds_{\upn} \XY_{b)} 
                         - 8 \kappa_n \sone_{b)}  - 4 \U_{b)f} P^{cf} 
                         \XY_c \right ) 
- 2\upn(\kappa_n) \ell_a \ell_b + \frac{2}{\n-1} \left ( - \upn (\kappa_n) \ll +  \upn(\ZY)  \right )  \gamma_{ab} + 2 \hatLn \YH_{ab},
    \nonumber                 
  \end{align*}
  where in the second equality we used Lemma \ref{hatLnP(V)Lemma}. Subtracting
  this from \eqref{t1}, the expression for $\R_{ab}$ becomes
\textcolor{black}{
  \begin{align*}
    \R_{ab} = & \Ricco_{(ab)} - \left ( 2\kappa_n + \mbox{tr}_P \bU \right ) \Y_{ab}
                + \nablao_{(a} \left ( \sone_{b)}  - 2 \kappa_n \ell_{b)}  \right )                     
- 2 \left ( \kappa_n^2  - \upn(\kappa_n) \right )  \ell_a \ell_b  \\
& + \ell_{(a} \left ( -2 \pounds_\upn   \XY_{b)} +
4 \kappa_n \sone_{b)}
                + 4 \kappa_n \XY_{b)}  \right )
+ \frac{2}{\n-1} \left (  \kappa_n \upn(\ll)  + \upn(\kappa_n) \ll
                - \upn (\ZY) + P^{cf} \nablao_c \XY_{f} \right ) \gamma_{ab} \\
& - 2 \XY_a \XY_b - 4 \XY_{(a} \sone_{b)}
    - \sone_a \sone_b
    - (\mbox{tr}_P \bY) \U_{ab}
    + 2 P^{cd} \U_{d(a} \F_{b)c} 
+    \O(\bm{\XY})_{ab} - 2 \hatLn \YH_{ab}.
  \end{align*}}
At this point we elaborate the term 
$\nablao_{(a} (2 \kappa_n \ell_{b)} )$. Taking into account \eqref{nablaoll}
and \eqref{nablaD}, this gives 
  \textcolor{black}{   \begin{align*}
\nablao_{(a} (2 \kappa_n \ell_{b)} )
= 2 \ell_{(a} D_{b)} \kappa_n + 2 \upn(\kappa_n) \ell_a \ell_b
- 2 \kappa_n \ll \U_{ab}.
\end{align*}}
Inserting this as well as the decomposition \eqref{decomY} the tensor
$\R_{ab}$ takes its nearly final form
\textcolor{black}{  \begin{align*}
    \R_{ab} = & \Ricco_{(ab)}  + \nablao_{(a} \sone_{b)} 
+ \kappa_n (\ZU)  \ell_a \ell_b   + \ell_{(a} \left ( -2 \pounds_\upn   \XY_{b)} - 2 D_{b)} \kappa_n +
4 \kappa_n \sone_{b)}
                - 2 (\ZU) \XY_{b)}  \right ) \\
& + \frac{1}{\n-1} \left (  2 \upn( \kappa_n \ll)  
                - 2\upn (\ZY) + 2 P^{cf} \nablao_c \XY_{f}
+ (\kappa_n \ll - \ZY) ( 2 \kappa_n + \ZU)
 \right ) \gamma_{ab} \\
& - 2 \XY_a \XY_b - 4 \XY_{(a} \sone_{b)}
    - \sone_a \sone_b
    + \left ( 2 \ll \kappa_n - \ZY   \right )  \U_{ab}
    + 2 P^{cd} \U_{d(a} \F_{b)c} \\ 
& +    \O(\bm{\XY})_{ab} - 2 \hatLn \YH_{ab} - \left ( 2 \kappa_n + \ZU \right )
\YH_{ab}
  \end{align*}}
  To arrive at the final result \eqref{decomConst} we first
apply Corollary \ref{decomprod} to
$\bm{\omega_1} = \bm{\omega_2} =  \bm{\XY} + \bsone$. Then, 
we  simply need to
replace 
\begin{align*}
\ell_a \ell_b = \tgamma_{ab} - \frac{\ll}{\n-1} \gamma_{ab}
\end{align*}
as well as the decomposition \eqref{decombU}, and simplify.
\end{proof}

\section{Covariant decomposition of the null shell equations}

\label{shells}

In this section we find an application of the decomposition obtained in Theorem
\ref{decomConstThm} in the context of null shells. The matching theory of spacetimes is the framework for constructing a new spacetime out of two  spacetimes that have been cut out along a hypersurface. More precisely, one considers two spacetimes with boundary $(\M^{\pm}, g^{\pm})$
and asks under which conditions one can construct a new spacetime by attaching their boundaries $\partial M^{\pm}$. The theory is well developed and has received many relevant contributions (see particularly \cite{Darmois, Lanczos1922, Lanczos1924, Lichnerowicz, Penrose1972, BonnorVickers, Israel, BarrabesIsrael, MarsSenovilla93,  MarsSenovillaVera, Senovilla2018}).

It turns out that one can formulate the matching theory at a completely detached level using the notion of hypersurface data. The basic idea was introduced in \cite{MarsGRG}  and it has been developed recently in the case when the
boundaries $\partial M^{\pm}$ are null in the papers \cite{Miguel1, Miguel2}. The outcome is that for the matching to be possible it is necessary and sufficient
that there exist two embeddings  $\Phi^{\pm} : \N \longrightarrow
M^{\pm}$ satisfying $\Phi^{\pm} (\N) = \partial M^{\pm}$, and two riggings
$\rig^{\pm}$ satisfying the orientation condition that one of them points inwards and the other outwards in their respective manifolds, such that the induced
metric hypersurface data on $\N$ are identical. The extrinsic tensor
$\bY$ induced from each embedding  (denoted by $\bY^{\pm}$ respectively) are in general different. Thus, any matching problem necessarily leads to a pair of hypersurface data
\begin{align*}
  \hypdatap, \qquad \hypdatam.
\end{align*}
It is a fact \cite{ClarkeDray, MarsSenovilla93},
that the matched spacetime admits a $C^1$ atlas where the metric is continuous
but in general no more regularity can be expected. In fact, 
there exists a $C^2$ atlas in the matched manifold where the metric is $C^1$
if and only if $\bY^+ = \bY^-$. The physical interpretation of the lack of differentiability of the  metric across the matching hypersurface is different in the null and in the non-null cases, and also depends on the gravity theory under consideration. Assuming General Relativity (see e.g. \cite{Senovilla1} and references therein for alternatives theories of gravitation) the physical meaning is as follows. When the hypersurfaces to be matched are non-null, the matching hypersurface carries energy and momentum which is associated to matter fields (by which we mean
{\em non-gravitational fields}) with support in the hypersurface. In the null case, the shell also carries energy and momentum, but now it can be associated both to the matter fields and/or to pure gravitational fields in the form of an impulsive gravitational wave.

Since the theme of this paper is null hypersurface data, we restrict ourselves to the null case from now on. Then 
a non-zero difference tensor
$[\bY] := \bY^+ - \bY^-$ corresponds to an infinitely thin concentration of matter and/or gravitational field.  The matter contents of the shell is described by
a symmetric, contravariant $(2,0)$ tensor $\Tau$ which is defined in terms of
the jump $[\bY]$ by means of
\begin{align*}
  \Tau = \tau(\epsilon[\bY]),
\end{align*}
where $\tau$ is the energy-momentum map (indeed, this is the reason for the name of the map) and $\epsilon$ is a sign that accounts for the relative orientation of the rigging vector with respect to the boundary. This sign is $\epsilon = +1$ if the rigging points from $M^-$ to $M^+$ and $\epsilon =-1$ if it points from
$M^-$ to $M^+$. In \cite{MarsGRG} the rigging was assumed to always point from
$M^-$ to $M^+$ so $\epsilon$ does not appear in that paper. The convenience to add this sign in the definition has been discussed recently in \cite{MiguelTesis,
ManzanoMars}.

The gauge behaviour of $\bY$ \eqref{gaugebY} implies at once that $[ \bY ]$ is  a field with gauge weight $q=1$. By Lemma \ref{weighttau} we conclude that
$\tau([Y])$ has gauge weight $-1$. The physical reason behind this behaviour is that
$\Tau$ is an energy-momentum tensor, so it describes energy or momentum
per unit volume (i.e. a density) and this notion depends on the measure of volume. Orientable metric hypersurface data admits \cite{MarsGRG,Mars2020} a canonical volume form $\bm{\eta}_{\ell}$ which turns out
to transform under a gauge as 
$\G_{(z,\gauge)} (\bm{\eta}_{\ell}) = |z| \bm{\eta}_{\ell})$. The product $\Tau \bm{\eta}_{\ell}$ is then gauge invariant, as it must, given that it
corresponds to a physical quantity describing energy and momentum inside the shell.

The fact that the energy-momentum map $\tau$ has a non-trivial kernel is responsible for the fact that a null shell can carry energy and momentum without carrying matter degrees of freedom. Whenever $[ \bY ]$ belongs to the kernel of
$\tau$, the shell describes a purely gravitational impulsive wave.
In the previous section we have obtained a covariant decomposition of the constraint tensor in terms of the decomposition obtained in Proposition \ref{Decom}. It is clear how to apply those results to the present setting. Each spacetime to be matched $(M^{\pm}, g^{\pm})$ has its own
Ricci tensor $\Ricc_{g^{\pm}}$, and their pull-back to $\N$ via
$\Phi^{\pm}$ define respective constraint tensors, denoted by
$\R^{\pm}$. The difference tensor $[\R] := \R^+ - \R^-$ is the jump of the tangential components of the Ricci tensor at each side (once the two spacetimes are
joined together). Via the Einstein field equations, they can be related to the energy-momentum contents of the bulk spacetime at each side, as well as to the jump of the cosmological constants on each domain in case they are assumed to be different.

The constraint tensor is expressible in terms of the hypersurface data, so taking its difference for the
corresponding data $\hypdatap$ and $\hypdatam$ provides an explicit link between
the jump $[\R]$ and the jump $[\bY]$, which one may call
{\em null shell equations} or also  {\em null Israel equations}.
In the case of non-null boundaries W. Israel \cite{Israel} obtained equations relating the jump
of the normal-normal and normal-tangential components of the Einstein tensor, i.e. $[\Ein(\nu,\cdot)]$, and the jump of the extrinsic curvatures at each side
$[K^{\nu}]$. They were generalized to the null case in \cite{BarrabesIsrael} and to the case of boundaries of arbitrary causal character in \cite{MarsGRG}. From the jump of the full constraint tensor $[\R]$ one can recover the equations of Israel because of the identities (we drop the
$\pm$ sign)
\begin{align*}
  \Ein_g(\nu,X) & \stackrel{\partial M}{=}
  \Ricc_g (\nu,X) \\
  \Ein_g(\nu,\rig) & \stackrel{\partial M}{=} \Ricc_g(\nu, \rig) -
  \left ( \Ricc_g(\nu,\rig) + \frac{1}{2} P^{ab} \Ricc_g (\Phi_{\star}(e_a),
  \Phi_{\star} (e_b)) \right ) g(\nu,\rig) \stackrel{\partial M}{=}
  - \frac{1}{2} P^{ab} (\Phi^{\star} \Ricc_g) (e_a,e_b).
    \end{align*}
Here $\nu$ is the normal to the null boundary satisfying
$g(\nu,\rig)=1$, $X$ is any tangential vector to the boundary and in the last equality we used the decomposition \eqref{contra}.
The left-hand sides of
these expressions involve only tangential components of the spacetime Ricci tensor, so they arise as contractions of the constraint tensor. Thus the null shell equations constitute an extension of the Israel equations.

To obtain the null shell equations we simply need to subtract the constraint tensors on both sides. The result is much simpler than the corresponding expressions on each side. The reason is that the metric part of the data is the same, so the difference of every tensor constructed solely from
$\{\gamma,\ellc,\elltwo\}$
is identically zero. This holds in particular for the
$\nablao$-Ricci curvature tensor $\Ricco$. We shall use the notation $[ \cdot ]$ to denote the difference between any geometric tensor constructed from the data
$\hypdatap$ and the corresponding tensor constructed from
$\hypdatam$. Given that in the previous section we have worked out
the constraint tensor in a decomposed form, we can also write down the null shells equations in a decomposed form.

Before stating the result, let us introduce some terminology. First we define
an abstract null shell as a pair of null hypersurface data  \cite{MarsGRG}.
Following \cite{MiguelTesis} we add an extra sign $\epsilon$ to account for the relative orientation of the riggings with respect to the boundaries 
\begin{definition}
  A {\bf null shell} is a a pair of null hypersurface data of the form
  $\hypdatapm$, where $\epsilon = \pm 1$.
\end{definition}
We have already said that the energy-momentum tensor of the shell is
$\Tau = \tau(\epsilon [\bY])$. In view of \eqref{decommu} and since $\upn$ corresponds (in the embedded case) to the null direction ruling the null hypersurface, it is reasonable
to interpret physically the various terms in the canonical decomposition \eqref{decomY} of $[\bY]$ as follows (see \cite{MiguelTesis} for a more detailed justification).
\begin{definition}
  \label{pressureetc}
  Let $\hypdatapm$ be a null shell. The {\bf energy-density} $\den$,
  the {\bf pressure} $\pre$, and the
  {\bf transverse energy flux} $\flux$ of the shell are defined by
  \begin{align*}
    \den := - \epsilon [ \trP \bY], \qquad \quad \textcolor{black}{\pre :=  \epsilon [\kappa_n]}, \qquad \quad \flux = \epsilon
        [\unl{\overline{Y}}].
  \end{align*}
  Moreover, the {\bf impulsive gravitational wave} of the shell is
  $\bGrav := \epsilon [ \bYH]$.
\end{definition}
By construction, the energy-momentum tensor of the shell decomposes as
(see \eqref{decommu})
    \begin{align}
    \Tau = \den \, \upn \otimes \upn + 2 \upn \otimes_s \flux
   + \pre \big ( P + 2 \ll \upn \otimes \upn \big ).
    \label{TauTensor}
  \end{align}
    We can link the jump of the constraint tensors $[\R]$ to the energy contents of the shell and to the impulsive gravitational wave of the shell. To write down the final result we use the following notation:  for any
  pair of quantities  $A^{\pm}$ defined on
  $\hypdatapm$ we let
    $\underlineb{A} := \frac{1}{2} (A^+ + A^-)$.
  \begin{proposition}
      \label{nullshellEqs}
      Let $\hypdatapm$ be a null shell. Then, the jump of the
        constraint tensor satisfies the following equation
                   \begin{align}
\epsilon       [\R_{ab}] = & 
      \ell_{(a} \left ( -2 \pounds_\upn \uwidehat{\flux}_{b)} - 2 (\ZU) \uwidehat{\flux}_{b)}
      - 2 D_{b)} \pre +
4 \sone_{b)} \pre 
\right ) 
+ \pre (\ZU)  \tgamma_{ab}            \label{final} \\
& + \frac{2}{\n-1} \Big (  \upn( \ll \pre  + \den)  
+ \trP ( \nablao \,  \uwidehat{\flux} 
- 2  \avY \otimes_s \uwidehat{\flux} - 2 \uwidehat{\flux} \otimes_s \bsone )
+ 2 \ll \underlineb{(\kappa_n)} \, \pre  \nonumber \\
& \hspace{15mm} 
- \underlineb{(\ZY)} \, \mu + \ZU \left (
                \ll \pre  + \mu \right )  
\Big )
\gamma_{ab} \nonumber \\
  &   - 2 \hatLn \Grav_{ab} +   \O(\uwidehat{\flux})_{ab} 
  - 4 \left (  \avY \otimes_s \uwidehat{\flux} \right )^{H}_{ab}
  - 4 \left ( \bsone \otimes_s \uwidehat{\flux} \right )^{H}_{ab}
  -  2 \underlineb{(\kappa_n)} \Grav_{ab} 
  - 2 \pre \underlineb{(\YH)_{ab}}
  \nonumber \\
&   - \ZU \,  \Grav_{ab} + \left ( \den + 2 \ll \pre  \right )  \UH_{ab}. 
\nonumber
          \end{align}
                  \end{proposition}
\begin{proof}
  Note first that, by linearity, the operators $\hatLn$ and $\O$ in
  \eqref{defhatLn} and \eqref{defOexp} satisfy
  \begin{align*}
    [\hatLn(\bY)] = \hatLn ([ \bY]), \qquad
    [\O(\bm{\XY})] = \O( [ \bm{\XY}]),
  \end{align*}
  and for any tensor $T$ and scalar $f$ we obviously also have $[\nablao T] = \nablao[T]$  and $[Df] = D[f]$ (cf. \eqref{nablaD}).
  We have already said that $[\Ricco] = [\bsone] = [\bU] = [\bF]=0$ and, more generally, any tensor defined only from the metric part of the data has vanishing jump. Taking difference in \eqref{decomConst} yields
\textcolor{black}{
  \begin{align}
      [\R_{ab}] = & 
 \ell_{(a} \left ( -2 \pounds_\upn   [\bm{\XY}]_{b)} - 2 D_{b)} [\kappa_n] +
4 \sone_{b)}[\kappa_n] 
                - 2 (\ZU) [\bm{\XY}]_{b)}  \right ) 
+ [\kappa_n] (\ZU)  \tgamma_{ab} 
\label{difference1} \\
& + \frac{2}{\n-1} \Big (  \upn( [\kappa_n] \ll - [\ZY])  
+ \trP ( \nablao \,  [\bm{\XY}]
- [ \bm{\XY} \otimes  \bm{\XY}] - 2 [\bm{\XY}] \otimes_s \bsone )
+ \ll [\kappa_n^2] \nonumber \\
& \hspace{15mm} - [ \kappa_n \ZY] + \ZU \left (
                \ll [\kappa_n]  - [\ZY] \right )  
\Big )
\gamma_{ab} \nonumber \\
  &   - 2 \hatLn [\YH]_{ab} +   \O([\bm{\XY}])_{ab} 
  - 2 \left [ \bm{\XY} \otimes \bm{\XY} \right ]^{H}_{ab}
  - 4 \left ( \bsone \otimes_s [\bm{\XY}] \right )^{H}_{ab}
   -  2 [\kappa_n \YH_{ab}] \nonumber \\
&   - \ZU 
   [\YH]_{ab} + \left ( 2 \ll [\kappa_n] -  [\ZY]  
   \right )  \UH_{ab}. \nonumber
    \end{align}}
    Now, for any pair of quantities $A,B$ the jump of their product is
    \begin{align*}
      [AB] = A^+ B^+ - A^- B^- =
      \frac{1}{2} ( A^+ + A^{-}  ) (B^+ - B^-)
      + \frac{1}{2} (A^+ - A^- ) (B^+ + B^- )
      =  \underlineb{A} [B] + \underlineb{B} [A].
    \end{align*}
     In particular, for a tensor $A$ on has $[ A \otimes A]
    = 2 \underlineb{A} \otimes_s [A]$.
    Note also that
\begin{align*}
    \uwidehat{J} := \gamma( J, \cdot ) =
    \epsilon \gamma( [\unl{\overline{Y}}] , \cdot )
  = \epsilon [ \gamma (\unl{\overline{Y}}  , \cdot ) ]
= \epsilon [ \bm{\XY}],
\end{align*}
where in the last equality we inserted \eqref{raise} and used
\eqref{prod4} together with the fact that
${\bm \XY} (\upn)=0$. Replacing this into \eqref{difference1} yields \eqref{final}, after using Definition \ref{pressureetc}. 
\end{proof}
In the embedded case, the jump of constraint tensor in the right hand side of \eqref{final} can be computed from the properties of each spacetime $(M^{\pm},g^{\pm})$. Since the constraint tensor is symmetric, it can also be decomposed according to Proposition \ref{Decom}.  We therefore define the quantities
\begin{align*}
  \bulkden := \R(\upn,\upn), \qquad \quad
  \bulkpre:= \trP \R, \qquad \quad
  \bulkflux := \R(\upn, \cdot ) - \bulkden \ellc, \qquad
  \quad \bulkGrav := (\I - \P) ( \R)
  \end{align*}
so that
\begin{align}
  \R = 2 \ellc \otimes_s  \bulkflux + \bulkden \tgamma
 + \frac{\bulkpre}{\n-1}  \gamma + \bulkGrav.
\label{DecomR}
\end{align}
  Given that the decomposition of any symmetric tensor according to Proposition
  \ref{Decom} is unique, the combination of \eqref{DecomR} and \eqref{final}
  can be split into its irreducible parts. The result is
  \begin{theorem}
    Let $\hypdatapm$ be a null shell.With the notation introduced above, the following expressions called
    {\bf null shell equations} hold true
  \textcolor{black}{   \begin{align*}
 \epsilon     [\bulkden] & = \pre \, \ZU, \\
 \epsilon     [\bulkflux] & = -2 \pounds_\upn \uwidehat{\flux} - 2 (\ZU) \uwidehat{\flux}
      - 2 D  \pre +
4 \pre \, \bsone,  \\ 
\epsilon [\bulkpre] & = 2 \upn( \ll \pre  + \den)  
+ 2 \trP ( \nablao \,  \uwidehat{\flux}  
- 2  \avY \otimes_s \uwidehat{\flux} - 2 \uwidehat{\flux} \otimes_s \bsone )
+ 4 \ll \underlineb{(\kappa_n)} \, \pre  
- 2 \underlineb{(\ZY)} \, \mu + 2 \ZU \left (
                \ll \pre  + \mu \right ), \\
       \epsilon [\bulkGrav] & =
       - 2 \hatLn \bGrav +   \O(\uwidehat{\flux}) 
  - 4 \big (  \avY \otimes_s \uwidehat{\flux} \big )^{H}
  - 4 \big ( \bsone \otimes_s \uwidehat{\flux} \big )^{H}
  -  2 \underlineb{(\kappa_n)} \bGrav
  - 2 \pre \, \underlineb{(\YH)}   - (\ZU) \,  \bGrav
  + \left ( \den + 2 \ll \pre  \right )  \UH. 
    \end{align*}}
\label{main}
  \end{theorem}

  These equations are  remarkable because they have a hierarchical structure.
  Assume that we know the geometry of the spacetimes to be matched, i.e. that we know the constraint tensors $\R$ for both hypersurface data. This means that the left-hand sides of the equations in Theorem \ref{main} are known. Then, the first equation determines $p$ (provided $\ZU$ is non-zero). With this information
  the second equation is a transport equation for $\uwidehat{\flux}$. Once this has been solved, the third equation gives a transport equation for $\den$. Finally, the last equation involves all the variables, but once the previous equations have already been solved, it becomes simply a transport equation
  for $\bGrav$. It is even more remarkable that the splitting has been achieved in  a fully covariant way  (as no coordinates have been introduced in
  $\N$), in a completely gauge covariant manner (the equations hold in any gauge) and without assuming a foliation of $\N$ by spacelike sections. In fact we have made no topological assumption whatsoever\footnote{Besides those needed for $(\N,\gamma)$ to be a ruled null manifold, of course (see Section
    \ref{Nullmanifolds}).}
  on $\N$
  so such a foliation need not exist.

  \section*{Acknowledgements}

  I am very grateful to Gabriel S\'anchez P\'erez for useful comments  on a previous version of the paper and Miguel Manzano for spotting several typos.
  This work has been supported by
PID2021-122938NB-I00 (Spanish
Ministerio de Ciencia, Innovaci\'on y Universidades and FEDER ``A way of making Europe'').

\end{document}